\documentclass[12pt]{article}

\usepackage{url,amsmath,amssymb,amstext,amsthm,graphicx}

\newenvironment{claimproof}[1][Proof of Claim]
{ \begin{proof}[#1]}
{\end{proof}  }
\newcommand{\p}[1]{\mathbb{P}\left[{#1}\right]}
\newcommand{\e}[1]{\mathbb{E}\left[{#1}\right]}
\newcommand{\eq}[1]{(\ref{#1})}
\newcommand{\eps}{\varepsilon}
\newcommand{\geo}{\operatorname{Geo}}

\newcommand{\rad}{\operatorname{ht}}
\newcommand{\wh}{\operatorname{wht}}

\newcommand{\outdeg}{\operatorname{out-deg}}

\theoremstyle{plain}
\newtheorem{theorem}{Theorem}
\newtheorem{lemma}[theorem]{Lemma}

\newtheorem{corollary}[theorem]{Corollary}

\theoremstyle{definition}
\newtheorem*{definition}{Definition}

\theoremstyle{remark}

\newtheorem*{claim}{Claim}

\title{It's a Small World for Random Surfers}

\author{Abbas Mehrabian \\
{\small Department of Combinatorics and Optimization, University of Waterloo} \\
{\small \texttt{amehrabi@uwaterloo.ca}}
\and Nick Wormald\thanks{Supported by Australian Laureate Fellowships grant FL120100125.}\\
{\small School of Mathematical Sciences, Monash University} \\
{\small \texttt{nick.wormald@monash.edu}}
}
\date{}
\begin{document}

\maketitle

\begin{abstract}
We prove logarithmic upper bounds for the diameters of the random-surfer Webgraph model and the PageRank-based selection Webgraph model, confirming the small world phenomenon holds for them.
In the special case when the generated graph is a tree, we provide close lower and upper bounds for the diameters of both models.

\end{abstract}

\section{Introduction}
Due to the ever growing interest in social networks, the Webgraph, biological networks, etc., in recent years a great deal of research has been built around modelling real world networks (see, e.g., the monographs~\cite{Bonato,Chakrabarti2012,complexgraphs,durrett}).  
One of the important observations about many real world networks involves the diameter, which is the maximum shortest-path distance between any two nodes. The so-called \emph{small world phenomenon} is that the diameter of a network is significantly smaller than its size, 
typically growing as a polylogarithmic function.

The Webgraph is a directed graph whose vertices are the static web pages, and there is an edge joining two vertices if there is a hyperlink in the first page pointing to the second page. Barab\'{a}si and Albert~\cite{barabasi_albert} in 1999 introduced one of the first models for the Webgraph, widely known as the \emph{preferential attachment model}. Their model can be informally described as follows (see~\cite{diameter_preferential_attachment} for the formal definition).
Let $d$ be a positive integer. We start with a fixed small graph, and in each time-step a new vertex appears and is joined to $d$ old vertices, where the probability of joining to each old vertex is proportional to its \emph{degree}.
Pandurangan, Raghavan and Upfal~\cite{pagerank_model_conference} in 2002 introduced the \emph{PageRank-based selection model} for the Webgraph. This model is similar to the previous model, except the attachment probabilities are proportional to the \emph{PageRanks} of the vertices rather than their degrees.
Blum, Chan, and Rwebangira~\cite{random_surfer} in 2006 introduced a \emph{random-surfer model} for the Webgraph, in which the $d$ out-neighbours of the new vertex are chosen by doing $d$ independent random walks that start from random vertices and whose lengths are geometric random variables with parameter $p$.
It was shown that under certain conditions, the previous two models are equivalent.
See Section~\ref{sec:def} for the formal definitions of these models, and the condition for their equivalence.

The directed models considered here generate directed acyclic graphs (new vertices create edges to old vertices), so it is natural to define the \emph{diameter} of a directed graph as the maximum shortest-path distance between any two vertices in its underlying undirected graph.
The diameter of the preferential attachment model was analysed by Bollob{\'a}s and Riordan~\cite{diameter_preferential_attachment}.
{Previous work on the PageRank-based selection and random-surfer models has focused on their degree distributions.} 
To the best of our knowledge, the diameters of {these} models have not been studied previously, and it is an open question even whether {they} have logarithmic diameter. 
One of the main contributions of this paper is giving logarithmic upper bounds for their diameters.
We also give close lower and upper bounds in the special case $d=1$, namely when the generated graph is (almost) a tree.
It turns out that the key parameter in this case is the \emph{height} of the generated random tree. 
We find the asymptotic value of the  height for all $p \in[0.21,1]$, and for $p \in(0, 0.21)$ we provide logarithmic lower and upper bounds. Our results hold \emph{asymptotically almost surely {(a.a.s.)}}, which means the probability that they are true approaches 1 as the number of vertices grows. 

\subsection{Our approach and organization of the paper}

In the preferential attachment model and most of its variations (see, e.g.,~\cite{barabasi_albert,diameters_pa,pa_general_1,pa_general_2})
the probability that the new vertex attaches to an old vertex $v$, called the \emph{attraction} of $v$, is proportional to a deterministic function of the degree of $v$.
In other variations (see, e.g.,~\cite{multiplicative_fitness_def,additive_fitness})
the attraction also depends on the so-called `fitness' of $v$, which is a random variable generated independently for each vertex and does not depend on the structure of the graph. For analysing such models when they generate trees, a typical technique is to approximate them with population-dependent branching processes and prove that results on the corresponding branching processes carry over to the original models. A classical example is  Pittel~\cite{random_recursive_trees} who estimated the height of random recursive trees.  Bhamidi~\cite{bhamidi} used this technique to show that the height of a variety of preferential attachment trees is asymptotic to a constant times the logarithm of the number of vertices, where the constant depends on the parameters of the model.

In the random-surfer Webgraph model, however, the attraction of a vertex does not depend only on its degree, but rather on the graph's general structure, so the branching processes techniques cannot apply directly, and new ideas are needed.

The  crucial  novel idea in our proof is to reduce the attachment rule to a simple one, with the help of introducing (possibly negative) `weights' for the edges. 
First, consider the general case, $d\ge 1$. 
Whenever a new vertex appears, it builds $d$ new edges to old vertices; suppose that we mark the first new edge.
Then the marked edges induce a spanning tree
whose diameter we bound, and thus we get an upper bound for the diameter of the random-surfer Webgraph model.

In the special case $d=1$, we obtain a \emph{random recursive tree} with edge weights, and then we  adapt a powerful technique developed by Broutin and Devroye~\cite{treeheight} (that uses branching processes) to study its weighted height. 
This technique is based on large deviations.
Their main theorem~\cite[Theorem~1]{treeheight} is not applicable here for two reasons. Firstly, the weights of edges on the path from the root to each vertex are not independent, and secondly, the weights can be negative.

We define the models and state our main results in Section~\ref{sec:def}. 
In Section~\ref{sec:graph} we give logarithmic upper bounds for the diameters of the random-surfer Webgraph model and the PageRank-based selection Webgraph model in the general case $d\ge 1$.
In Sections~\ref{sec:transform}--\ref{sec:upper} we focus on the special case $d=1$ and prove close lower and upper bounds for the heights and diameters of the models.
Section~\ref{sec:transform} contains the main technical contribution of this paper, where we explain how to transform the random-surfer tree model into one that is easier to analyse. 
The lower and upper bounds are proved in Sections~\ref{sec:lower} and~\ref{sec:upper}, respectively. Concluding remarks appear in Section~\ref{sec:conclude}.
For easing the flow of reading the paper, proofs of some technical lemmas has been put in the appendix.

\section{Definitions and main results}
\label{sec:def}

Given $p\in(0,1]$, let $\geo(p)$ denote a geometric random variable with parameter $p$; namely for every nonnegative integer $k$,
$\p{\geo(p) = k} = (1-p)^kp$.

\begin{definition}[Random-Surfer Webgraph model~\cite{random_surfer}]
Let $d$ be a positive integer and let $p\in (0,1]$.
Generate a random directed rooted $n$-vertex multigraph, with all vertices having out-degree $d$.
Start with a single vertex $v_0$, the root, with $d$ self-loops.
At each subsequent step $s$, where
 $1 \le s \le n-1$, a new vertex $v_s$ appears and $d$ edges are created from it to vertices in $\{v_0,v_1,\dots,v_{s-1}\}$, by doing the following probabilistic procedure $d$ times, independently:
choose a vertex $u$ uniformly at random from $\{v_0,v_1,\dots,v_{s-1}\}$,
and a fresh random variable $X=\geo(p)$;
perform a simple random walk of length $X$ starting from $u$, and join $v_s$ to the last vertex of the walk. 
\end{definition}

The motivation behind this definition is as follows.
Think of the vertex $v_s$ as a new web page that is being set up.
Say the owner wants to put $d$ links in her web page.
To build each link, she does the following: she goes to a random page. With probability $p$ she likes the page and puts a link to that page.
Otherwise, she clicks on a random link on that page, and follows the link to a new page.
Again, with probability $p$ she likes the new page and puts a link to that, 
otherwise clicks on a random link etc., until she finds a desirable page to link to.
The geometric random variables correspond to this selection process.

Our main result regarding the diameter of the random-surfer Webgraph model is the following theorem
(recall that the diameter of a directed graph is defined as the diameter of its underlying undirected graph).
All logarithms are natural in this paper.

\begin{theorem}
\label{thm:diameter_webgraph_new}
Let $d$ be a positive integer and let $p\in(0,1]$. 
A.a.s.\ as $n\to \infty$
the diameter of the random-surfer Webgraph model with parameters $p$ and $d$
is at most $8 e^p (\log n) / p$.
\end{theorem}

{
Notice that the upper bound in Theorem~\ref{thm:diameter_webgraph_new} does not depend on $d$ (whereas one would expect that the diameter must decrease asymptotically as $d$ increases).
This independence is because in our argument we employ only the first edge created by each new vertex to bound the diameter.}

{
When $d=1$, we show in Theorem~\ref{thm:diameter} below that the diameter is a.a.s.\ $\Theta(\log n)$.
An interesting open problem is to evaluate the asymptotic value of the diameter when $d>1$. 
In this regime the diameter might be of a smaller order, e.g.\ $\Theta(\log n / \log \log n)$, as is the case for the preferential attachment model (see~\cite[Theorem~1]{diameter_preferential_attachment}).
}

A \emph{random-surfer tree} is an undirected tree obtained from a random-surfer Webgraph with $d=1$ by deleting the self-loops of the root and ignoring the edge directions.
The \emph{height} of a tree is defined as the maximum graph distance between a vertex and the root.
Our main result regarding the height of the random-surfer tree model is the following theorem.

\begin{theorem}
\label{thm:main}
For $p\in (0,1)$, let $s{=s(p)}$ be the unique solution in $(0,1)$ to
\begin{equation}
\label{s_def}
s \log \left ( \frac{(1-p)(2-s)}{1-s} \right) = 1 \:.
\end{equation}
Let $p_0 \approx 0.206$ be the unique solution in $(0,1/2)$ to
\begin{equation}
\label{p0_def}
\log \left( \frac{1-p}{p} \right) = \frac{1-p}{1-2p} \:.
\end{equation}
Define the functions $c_L, c_U : (0,1) \to \mathbb{R}$ as
\begin{equation*}
c_L (p) = \exp(1/s) s (2-s) p \:,
\end{equation*}
and
\begin{equation*}
c_U (p) =
\begin{cases}
c_L (p) & \mathrm{if\ } p_0 \le p < 1 \\
\left( \log \left( \frac{1-p}{p} \right) \right)^{-1} & \mathrm{if\ } 0 < p < p_0 \:.
\end{cases}
\end{equation*}
For every fixed $\eps>0$, a.a.s.\ as $n\to \infty$
the height of the random-surfer tree model with parameter $p$
is between $(c_L(p) - \eps) \log n$ and $(c_U(p) + \eps) \log n$.
\end{theorem}

\begin{figure}
\begin{center}
\includegraphics[scale=0.3]{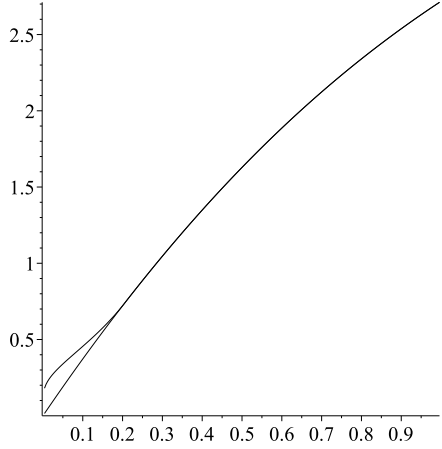}
\end{center}
\caption{The functions $c_L$ and $c_U$ in Theorems~\ref{thm:main} and~\ref{thm:diameter}.}
\label{fig:plot}
\end{figure}

The value $p_0$ and the functions $c_L$ and $c_U$ (plotted in Figure~\ref{fig:plot}) are well defined by Lemma~\ref{lem:uniques} below.
Also, $c_L$ and $c_U$ are continuous, and
$\lim_{p\to 0} c_L(p) = \lim_{p\to 0} c_U(p) = 0$ and
$\lim_{p\to 1} c_L(p)  = e$.
We suspect that the gap between our bounds when $p<p_0$ is an artefact of our proof technique, and we do not expect a phase transition in the behaviour of the height at $p=p_0$.

We also prove lower and upper bounds for the diameter, which are close to being tight.

\begin{theorem}
\label{thm:diameter}
Let $c_L$ and $c_U$ be defined as in Theorem~\ref{thm:main}.
For every fixed $\eps>0$, a.a.s.\ as $n\to \infty$
the diameter of the random-surfer tree model with parameter $p\in(0,1)$
is between $(2 c_L(p) - \eps) \log n$ and $(2 c_U(p) + \eps) \log n$.
\end{theorem}
Immediately, we have the following corollary.
\begin{corollary}
Let $c_L$ and $p_0$  be defined as in Theorem~\ref{thm:main}.
For any  $p\in[p_0, 1)$, the height of the random-surfer tree model with parameter $p$
is a.a.s.\ asymptotic to $c_L(p) \log n$  as $n\to \infty$, and its diameter is a.a.s.\ asymptotic to $2c_L(p) \log n$.
\end{corollary}

{
A natural open problem is to close the gap between the lower and upper bounds in Theorems~\ref{thm:main} and~\ref{thm:diameter}
when $p < p_0$. It seems that for solving this problem new ideas are required.
}

We now define the PageRank-based selection model introduced in~\cite{pagerank_model_conference,pagerank_model_journal}.
\begin{definition}[PageRank and the PageRank-based selection Webgraph model~\cite{pagerank_model_conference,pagerank_model_journal}]
Let $d$ be a positive integer and let $p,\beta\in[0,1]$.
The \emph{PageRank} of a directed graph is a probability distribution over its vertices, which is the stationary distribution of the following random walk. The random walk starts from a vertex chosen uniformly at random. In each step, with probability $p$ it jumps to a vertex chosen uniformly at random, and with probability $1-p$ it walks to a random out-neighbour of the current vertex.

The PageRank-based selection Webgraph model is
a random $n$-vertex directed multigraph with all vertices having out-degree $d$, {generated as follows.}
{It starts} with a single vertex with $d$ self-loops. 
At each subsequent step a new vertex appears, chooses $d$ old vertices and attaches to them (where a vertex can be chosen multiple times).
These choices are independent and the 
head of each edge is a uniformly random vertex with probability $\beta$,
and is a vertex chosen according to the PageRank distribution with probability $1-\beta$.
\end{definition}
The motivation behind this definition is as follows.
Consider the case $\beta=0$.
Think of the vertex $v_s$ as a new web page that is being set up.
Say the owner wants to put $d$ links in her web page.
She finds the destination pages using $d$ independent Google searches.
Since Google sorts the search results according to their PageRank (see~\cite{pagerank_def}), 
the probability that a given page is linked to is close to its PageRank.

Our main result regarding the diameter of the PageRank-based selection model is the following theorem.

\begin{theorem}
\label{thm:diameter_pagerank}
Let $d$ be a positive integer and let $p,\beta\in(0,1]$. 
A.a.s.\ as $n\to \infty$
the diameter of the PageRank-based selection Webgraph model with parameters $d$, $p$, and $\beta$
is at most $8 e^p (\log n) / p$.
\end{theorem}

Chebolu and Melsted~\cite[Theorem~1.1]{pagerank_random_surfer} showed  the random-surfer Webgraph model is equivalent to the
PageRank-based selection Webgraph model with $\beta=0$ {(this fact also follows from Lemma~\ref{lem:cm} in Section~\ref{sec:graph})}.
Hence Theorems~\ref{thm:diameter_webgraph_new} follows immediately from Theorem~\ref{thm:diameter_pagerank}.
Moreover, the conclusions of Theorems~\ref{thm:main} and~\ref{thm:diameter} apply to 
the PageRank-based selection Webgraph model with $\beta=0$ and $d=1$.

In Theorems~\ref{thm:main} and~\ref{thm:diameter} we have assumed that $p<1$, since the situation for $p=1$ has been clarified in previous work.
Let $p=1$. Then a random-surfer tree has the same distribution as a so-called random recursive tree, the height of which is a.a.s.\ asymptotic to $e \log n$ as proved by Pittel~\cite{random_recursive_trees}. It is not hard to alter the argument in~\cite{random_recursive_trees} to prove that the diameter is a.a.s.\ asymptotic to $2e\log n$.
The diameter of a random-surfer Webgraph thus has also an asymptotically almost sure upper bound of $2e \log n$.
For the rest of the paper, we fix $p\in(0,1)$.

We include some definitions here.
Define the \emph{depth} of a vertex as the length of a shortest path (ignoring edge directions) connecting the vertex to the root,
and the \emph{height} of a graph $G$, denoted by $\rad(G)$, as the maximum depth of its vertices.
Clearly the diameter is at most twice the height.
In a weighted tree (a tree whose \emph{edges} are weighted), define the \emph{weight} of a vertex to be the sum of the weights of the edges connecting the vertex to the root, and the \emph{weighted height} of tree $T$, written $\wh(T)$, to be the maximum weight of its vertices.
We view an unweighted tree as a weighted tree with unit edge weights, in which case the weight of a vertex is its depth, and the notion of weighted height is the same as the usual height.

We will need two large deviation inequalities,
whose proofs are standard and can be found in
the appendix.

Define the function $\Upsilon:(0,\infty)\to\mathbb{R}$ as
\begin{equation}
\label{upsilon_def}
\Upsilon(x) = \begin{cases}
x - 1 - \log (x) & \mathrm{\ if\ }0 < x \le 1 \\
0 & \mathrm{\ if\ } 1 < x \:.
\end{cases}
\end{equation}
\begin{lemma}
\label{lem_gamma}
Let $E_1,E_2,\dots,E_m$ be independent exponential random variables with mean 1.
For any fixed $x>0$, as $m\to \infty$ we have
$$
\exp \left(-\Upsilon(x) m - o(m)\right) \le
\p{E_1 + E_2 + \dots + E_m \le xm}
\le
\exp (-\Upsilon(x) m) \:.
$$
\end{lemma}

Define the function $f : (-\infty,1] \to \mathbb{R}$ as
\begin{equation}
\label{f_def}
f(x) = (2-x)^{2-x} p (1-p)^{1-x} (1-x)^{x - 1} \:.
\end{equation}
\begin{lemma}
\label{pro:geo}
Let $Z_1,Z_2,\dots,Z_m$ be independent $ 1 + \geo(p)$ random variables, and let $\kappa \ge 1/p$.
Then we have
$\p{Z_1 + Z_2 + \dots + Z_m \ge \kappa m}
\le f(2-\kappa)^m$.
\end{lemma}

\section{Upper bound for the PageRank-based model}
\label{sec:graph}
In this section we prove Theorem~\ref{thm:diameter_pagerank}, which gives an upper bound for the diameter of the PageRank-based selection Webgraph model.
Theorem~\ref{thm:diameter_webgraph_new} follows immediately using~\cite[Theorem~1.1]{pagerank_random_surfer}.
We need a technical lemma, whose proof can be found in the appendix.

\begin{lemma}
\label{lem:technical}
Let $\eta,c$ be positive numbers satisfying $\eta \ge 4e^p / p$ and $c\le p \eta$. Then we have
$-c \Upsilon(1/c) + c \log f ( 2 - \eta / c) < {\max\{\eta(1-p)\log(1-p^3),
-0.15p\eta\}}
-1$.
\end{lemma}

We now describe an alternative way to generate the edge destinations in the PageRank-based selection model.
Define the non-negative random variable $\mathcal{L}$ as 
$$
\mathcal{L} = \mathcal{L}(p,\beta)=\begin{cases}
0 & \mathrm{with\ probability\ } \beta \:,\\
\geo(p) & \mathrm{with\ probability\ } 1-\beta \:. 
\end{cases}
$$
Note that $\geo(p)$ stochastically dominates $\mathcal{L}$.

\begin{lemma}
\label{lem:cm}
The head of each new edge in
the PageRank-based selection model can be obtained by sampling a vertex $u$ uniformly from the existing graph 
and performing a simple random walk of length $\mathcal{L}$ starting from $u$.
\end{lemma}
The proof is a straightforward generalization of {that of}~\cite[Theorem~1.1]{pagerank_random_surfer}.
\begin{proof}
Let $G$ denote the existing graph, and let $\pi:V(G)\to[0,1]$ denote the PageRank distribution.
Then by definition, $\pi$ is the unique probability distribution satisfying
\begin{equation}
\label{eq:pagerank}
\pi (v) = \frac{p}{|V(G)|} + (1-p) \sum_{u\in V(G)} \frac{\pi(u) \cdot \#(uv)}{\outdeg(u)} \:.
\end{equation}
Here $\#(uv)$ denotes the number of copies of the directed edge $uv$ in the graph (which is zero if there is no edge from $u$ to $v$), and
$\outdeg(u)$ denotes the out-degree of $u$.

It suffices to show that if we sample a vertex uniformly and perform a random walk of length $\geo(p)$, the last vertex of the walk has distribution $\pi$.
Let ${\tau}: V(G) \to [0,1]$ denote the probability distribution of the last vertex, 
let $\mathcal{P}$ denote the probability transition matrix of the simple random walk, and let $\sigma = \big[1/|V(G)|,1/|V(G)|,\dots,1/|V(G)|\big]^T$
be the uniform distribution.
Then we have
$$
\tau 
= \sum_{k=0}^{\infty} (1-p)^k p \mathcal{P}^k \sigma
= p \sigma + 
(1-p)\mathcal{P}
\left(\sum_{k=1}^{\infty} (1-p)^{k-1} p \mathcal{P}^{k-1} \sigma\right)
=
p \sigma
+
(1-p)\mathcal{P} \tau \:.
$$
Comparing with (\ref{eq:pagerank}) and noting that the stationary distribution of an ergodic Markov chain is unique, we find that $\tau=\pi$, as required.
\end{proof}

We now have the ingredients to prove 
Theorem~\ref{thm:diameter_pagerank}.

\begin{proof}[Proof of Theorem~\ref{thm:diameter_pagerank}.]
Let {$\eta = 4e^p / p$}.
We define an {auxiliary tree} whose node set equals the vertex set of the graph generated by the PageRank-based selection Webgraph model, and whose weighted height dominates the height of this graph.
Then we show a.a.s.\ this tree has weighted height at most $\eta \log n$, which completes the proof.

Initially the tree has just one vertex $v_0$.
By Lemma~\ref{lem:cm}, the growth of the PageRank-based selection model at each subsequent step $s \in \{1,2,\dots,n-1\}$ can be described as follows: 
a new vertex $v_s$ appears and $d$ edges are created from it to vertices in $\{v_0,v_1,\dots,v_{s-1}\}$, by doing the following probabilistic procedure $d$ times, independently:
choose a vertex $u$ uniformly at random from $\{v_0,v_1,\dots,v_{s-1}\}$,
and a fresh random variable $\mathcal{L}$;
perform a simple random walk of length $\mathcal{L}$ starting from $u$, and join $v_s$ to the last vertex of the walk.

Consider a step $s$ and the first chosen $u\in\{v_0,\dots,v_{s-1}\}$ and $\mathcal{L}$.
In the tree, we join the vertex $v_s$ to $u$ and set the weight of the edge $v_s u$ to be $\mathcal{L}+1$.
Note that the edge weights are mutually independent.
Clearly, the weight of $v_s$ in the auxiliary tree is greater than or equal to the depth of $v_s$ in the graph.
Hence, it suffices to show that a.a.s.\ the weighted height of the auxiliary tree is at most $\eta \log n$.
We work with the tree in the rest of the proof.

Let us consider an alternative way to grow the tree, used by
Devroye, Fawzi, and Fraiman~\cite{biased_tree}, which results in the same distribution.
Let $U_1,U_2,\dots$ be i.i.d.\ uniform random variables in $(0,1)$.
Then for each new vertex $v_s$, we attach it to the vertex 
$v_{\lfloor s U_s \rfloor}$,
which is indeed a vertex uniformly chosen from 
$\{v_0,\dots,v_{s-1}\}$.

For convenience, we consider the tree when it has $n+1$ vertices $v_0,v_1,\dots,v_n$.
Let $D(s),W(s)$ denote the depth and the weight of vertex $v_{s}$, respectively.
We have
\begin{align*}
\p{\wh(\mathrm{auxiliary\ tree}) > \eta \log n}
& \le \sum_{s=1}^{n} \p{W(s) > \eta \log n} \\
& \le n \p{W(n) > \eta \log n} = 
\sum_{d=1}^{n} \mathcal{A}(d) \:,
\end{align*}
where we define
$$\mathcal{A}(d) = n \p{D(n)=d}
\p{W(n) > \eta \log n | D(n)=d}\:.$$
To complete the proof it is enough to show 
$\sum_{d=1}^{n} \mathcal{A}(d)=o(1)$.

Let $P(0)=0$ and for $s=1,\dots,n$, let $P(s)$ denote the index of the parent of $v_{s}$.
We have
$$\p{D(n)\ge d}
= \p{D(P(n)) \ge d-1}
= \dots
= \p{D(P^{d-1}(n)) \ge 1}
= \p{P^{d-1}(n) \ge 1} \:.
$$
Since $P(m) = \lfloor m U_m \rfloor \le m U_m$ for each $0\le m \le n$ and since the $U_i$ are i.i.d., we have
$$
\p{P^{d-1}(n) \ge 1}
\le \p{n U_1 U_2 \dots U_{d-1} \ge 1}\:.
$$
Let $E_i = - \log U_i$. Then $E_i$ is exponential with mean 1, and moreover,
\begin{align}
\p{D(n)\ge d}
& \le \p{n U_1 U_2 \dots U_{d-1} \ge 1} \nonumber\\
& = \p {E_1 + \dots + E_{d-1} \le \log n} 
\le \exp\left(-(d-1) \Upsilon \left(\frac{\log n}{d-1}\right)\right)\label{first}
\:,
\end{align}
where we have used Lemma~\ref{lem_gamma}.
The right-hand side is $o(1/n)$ for $d = 1.1e \log n$.
Hence to complete the proof we need only show that
\begin{align}
\mathcal{A}(d) = o(1/\log n ) \qquad \mathrm{\forall} d \in(0, 1.1 e \log n) \:.
\label{nos}
\end{align}

Fix an arbitrary positive integer $d \in(0, 1.1 e \log n)$.
The random variable $W(n)$, conditional on $D(n)=d$,  is a sum of $d$ i.i.d.\ $1+\mathcal{L}$ random variables.
Since $\geo(p)$ stochastically dominates $\mathcal{L}$, by Lemma~\ref{pro:geo} and since {$\eta > 1.1e / p$}, we have
\begin{equation}
\p{W(n) > \eta\log n | D(n) = d}
\le 
f(2-\eta \log n / d)^d \:,\label{second}
\end{equation}
where $f$ is defined in \eq{f_def}.

Combining \eq{first} and \eq{second}, we get
\begin{align}
\mathcal{A}(d)\le
\exp\Big[
\log n 
-(d-1) \Upsilon \left(\frac{\log n}{d-1}\right)
+d \log f(2-\eta \log n / d)
\Big] \:. \label{aali}
\end{align}

{
Let $c = d / \log n$ and $c_1 = c - 1/ \log n$.
Let $\vartheta=\max\{\eta(1-p)\log(1-p^3),
-0.15p\eta\}$.
Note that $\vartheta$ is a negative constant.
By Lemma~\ref{lem:technical}
and since the function $c \Upsilon(1/c)$ is uniformly continuous on $[0,1.1e]$, we find that for large enough $n$,
$$-c_1 \Upsilon(1/c_1) + c \log f ( 2 - \eta / c) < 
\vartheta/2 - 1 \:.
$$
Together with \eq{aali}, this gives
$
\mathcal{A}(d)\le\exp(\vartheta \log n /2)
$,
and \eq{nos} follows.}
\end{proof}

\section{Transformations of the random-surfer tree model}
\label{sec:transform}
In Sections~\ref{sec:transform}--\ref{sec:upper} we study the random-surfer tree model.
In this section we show how to transform the random-surfer tree model three times to eventually obtain a new random tree model, which we analyse in subsequent sections. The first transformation is novel. The second one was perhaps first used by Broutin and Devroye~\cite{treeheight}, and the third one probably by Pittel~\cite{random_recursive_trees}.

Let us call the random-surfer tree model the \emph{first model}.
First, we will replace the attachment rule with a simpler one by introducing \emph{weights} for the edges.
In the first model, the edges are unweighted and in every step $s$ a new vertex $v_s$ {appears}, chooses an old vertex $u$, and attaches to a vertex in the path connecting $u$ to the root, according to some rule.
We introduce a second model that is weighted, and such that there is a one to one correspondence between the vertices in the second model and in the first model. For a vertex $v$ in the first model, we denote its corresponding vertex in the second model by $\overline{v}$.
In the second model,  in every step $s$ a new vertex $\overline{v_s}$ {appears}, chooses an old vertex $\overline{u}$ and attaches to $\overline{u}$, and the weight $w(\overline{u}\:\overline{v_s})$ of the new edge $\overline{u}\:\overline{v_s}$ is chosen such that the \emph{weight} of $\overline{v_s}$ equals the \emph{depth} of $v_s$ in the first model. Let $w\left(\overline{u}\right)$ denote the weight of vertex $\overline{u}$.
Then it {follows from the definition of the random-surfer tree} model that $w(\overline{u}\:\overline{v_s})$ is distributed as $\max \{1 - \geo(p), 1 - w\left(\overline{u}\right) \}$.
{The} term $1 - w\left(\overline{u}\right)$ appears {here} solely because the weight of $\overline{v_s}$ is at least 1 (in the first model, the depth of $v_s$ is at least 1, since it cannot attach to a vertex higher than the root).
Because the depth of $v$ in the first model equals the weight of $\overline{v}$ in the second model, the height of the first model equals the weighted height of the second model.

We will need to make the degrees of the tree bounded, so we define a third model. In this model, the new vertex can attach just to the leaves. In step $s$ a new vertex $v_s$ appears, chooses a random leaf $u$ and joins to $u$ using an edge with weight distributed as $\max \{1 - \geo(p), 1 - w\left(\overline{u}\right) \}$. Simultaneously, a new vertex $u'$ appears and joins to $u$ using an edge with weight 0. Then we have $w(u) = w(u')$ and henceforth $u'$ plays the role of $u$, i.e.\ the  next vertex wanting to attach to $u$, but cannot do so because $u$ is no longer a leaf, may attach to $u'$ instead.  Clearly there exists a coupling between the second and third models in which the weighted height of the third model, when it has $2n-1$ vertices, equals the weighted height of the second model with $n$ vertices.
In fact the second model may be obtained from the third one by contracting all zero-weight edges.
We can thus study the weighted height of the first model by studying it in the third model.

All the above models were defined using discrete time steps.
We now define a fourth model using the following continuous time branching process, which we call $\mathcal{P}$.
At time 0 the root is born. From this moment onwards, whenever a new vertex $v$ is born (say at time $\kappa$), it waits for a random time $E$, which is distributed exponentially with mean 1, and after time $E$ has passed (namely, at absolute time $\kappa + E$) gives birth to two children $v_1$ and $v_2$, and dies. The weights of the edges $vv_1$ and $vv_2$ are generated as follows: vertex $v$ chooses $i\in\{1,2\}$ independently and uniformly at random. The weight of $vv_i$ is distributed as $\max \{1 - \geo(p), 1 - w\left(v\right) \}$ and the weight of $vv_{3-i}$ is 0.
 Given $t\ge 0$, we denote by $T_t$ the almost surely finite random tree obtained by taking a snapshot of this process at time $t$. By the memorylessness of the exponential distribution, if one starts looking at this process at any deterministic moment, the next leaf to die is chosen uniformly at random.
Hence for any stopping time $\tau$, the distribution of $T_{\tau}$, conditional on $T_{\tau}$ having $2n-1$ vertices, is the same as the distribution of the third model when it has $2n-1$ vertices.

The following lemma implies that certain results for $T_t$ carry over to results for the random-surfer tree model.

\begin{lemma}
\label{lem:equal_logs}
Assume that there exist constants $\theta_L,\theta_U$ such that for every fixed $\eps>0$,
$$\p{\theta_L(1-\eps) t \le \wh(T_t) \le \theta_U(1+\eps) t} \to 1 $$
as $t\to\infty$.
Then for every fixed $\eps>0$, a.a.s.\ as $n\to\infty$ the height of the random-surfer tree model is between
$\theta_L(1-\eps) \log n$ and $\theta_U(1+\eps) \log n$.
\end{lemma}

\begin{proof}
Let $\ell_n = 2n - 1$, and let $\eps>0$ be fixed.
For the process $\mathcal{P}$, we define three stopping times as follows:
\begin{description}
\item $a_1$ is the deterministic time $(1-\eps) \log (\ell_n)$.
\item $A_2$ is the random time when the evolving tree has exactly $\ell_n$ vertices.
\item $a_3$ is the deterministic time $(1+\eps) \log (\ell_n)$.
\end{description}

By  hypothesis, a.a.s.\ as $n\to\infty$ we have
\begin{equation}
\label{byassumption}
(1-\eps) \theta_L  \log (\ell_n) \le \wh\left(T_{a_1}\right) \mathrm{\ and\ }
\wh\left(T_{a_3}\right) \le (1+\eps) \theta_U \log (\ell_n) \:.
\end{equation}

Broutin and Devroye~\cite[Proposition~2]{treeheight}  considered the infinite process $T_t$ as $t\to\infty$ and proved that almost surely
$${\lim_{t\to\infty} \frac{\log|V({T}_t)|}{t} =1}\:,$$
which implies that a.a.s.\  as $t\to\infty$, we have ${\log |V({T}_t)| \sim t}$.
This means that, as $n\to\infty$, a.a.s.
$$\log |V({T}_{a_1})| \sim a_1 = (1 - \eps) \log (\ell_n) \:,$$
and hence  $|V({T}_{a_1})| < \ell_n$, which implies $a_1 < A_2$.
Symmetrically, it can be proved that a.a.s.\ as $n\to\infty$ we have $A_2 < a_3$.
It follows that a.a.s.\ as $n\to\infty$
\begin{equation}
\label{sandwich}
\wh\left(T_{a_1}\right) \leq \wh \left(T_{A_2}\right) \leq \wh \left(T_{a_3}\right) \:.
\end{equation}
On the other hand, as noted above,
$T_{A_2}$ has the same distribution as the third model with $2n-1$ vertices,
 whose weighted height  has the same distribution as  that of  the random-surfer tree model with $n$ vertices.
Chaining \eq{byassumption} and \eq{sandwich}
completes the proof.
\end{proof}

It will be convenient to define $T_t$ in a static way, which is equivalent to the dynamic definition above.

\begin{definition}[$T_{\infty}, T_t$]
Let $T_{\infty}$ denote an infinite binary tree.
To every edge $e$ is associated a random vector $(E_e, W_e)$ and to every vertex $v$ a random variable $W_v$, where the $W_e$'s and $W_v$'s are the \emph{weights}.
The law for $\{E_e\}_{e\in E(T)}$ is easy: first with every vertex $v$ we associate independently an exponential random variable with mean 1, and we let the values of $E$ on the edges joining $v$ to its two children be equal to this variable.
In the dynamic interpretation, this random variable denotes the length of life of $v$.
Generation of the weights is done in a top-down manner,
where we think of the root as the top vertex.
Let the weight of the root be zero.
Let $v$ be a vertex whose weight has been determined, and let $v_1,v_2$ be its two children.
Choose $i\in\{1,2\}$ independently and uniformly at random,
and then choose $Y = 1 - \geo(p)$ independently of previous choices.
Then let
\begin{equation}
\label{weights_T}
W_{vv_i} = \max \{Y, 1 - W_v \},\quad
W_{v_i} = W_v + W_{vv_i} \:,
\end{equation}
and
$$W_{vv_j} = 0,\quad
 W_{v_j} = W_v$$
for $j = 3-i$.

For a vertex $v$, let $\pi(v)$ be the set of edges of the unique path connecting $v$ to the root.
It is easy to check that the weight of any vertex $v$ equals $\sum_{e\in\pi(v)} W_e$.
We define the \emph{birth time} of a vertex $v$, written $B_v$, as
$$B_v = \sum_{e\in\pi(v)} E_e \:,$$
where the birth time of the root is defined as zero.
Finally, given $t\ge 0$ we define $T_t$ as the subtree of $T_{\infty}$ induced by vertices with birth time at most $t$.
Note that $T_t$ is finite almost surely.
\end{definition}

\section{Lower {bounds} for the random-surfer tree model}
\label{sec:lower}
Here we prove the lower {bounds in Theorems~\ref{thm:main} and~\ref{thm:diameter}}.
For this,  we consider another infinite binary tree $T'_{\infty}$
which is very similar to $T_{\infty}$, except for the generation rules for the weights, which are as follows.
Let the weight of the root be zero.
Let $v$ be a vertex whose weight has been determined, and let $v_1,v_2$ be its two children.
Choose $i\in\{1,2\}$ independently and uniformly at random,
and choose $Y = 1 - \geo(p)$ independently of previous choices.
Then let
\begin{equation}
\label{weights_T'}
W_{vv_i} = Y \mathrm{\ and\ } W_{v_i} = W_v + W_{vv_i}
\end{equation}
and
$$W_{vv_j} = 0\mathrm{\ and\ } W_{v_j} = W_v$$
for $j = 3-i$.
Comparing  \eq{weights_T'}  {with} \eq{weights_T}, we find that the weight of every vertex in $T'_{\infty}$
is stochastically less than or equal to that of its corresponding vertex in $T_{\infty}$.
The tree $T'_t$ is defined as before.
Clearly probabilistic lower bounds for $\wh(T'_t)$ are also probabilistic lower bounds for $\wh(T_t)$. Distinct vertices $u$ and $v$ in a tree are called \emph{antipodal} if the unique $(u,v$)-path in the tree passes through the root.

\begin{lemma}
\label{lem_bd_lower}
Consider the tree $T'_{\infty}$.
Let $\gamma_L:(0,1)\to\mathbb{R}$ be such that for every $a\in (0,1)$, each vertex $u$ and each descendent $v$ of $u$ that is $m$ levels deeper,
\begin{equation}
\label{lower_bound_condition}
\p{W_v - W_u \ge a m} \ge \exp (-m \gamma_L (a) - o(m))
\end{equation}
as $m\to \infty$.
Assume that there exist $\alpha^*,\rho^* \in(0,1)$ with
\begin{equation}
\label{alpharho_lower}
\gamma_L(\alpha^*) + \Upsilon(\rho^*) = \log 2 \:.
\end{equation}
Then for every fixed $\eps>0$, a.a.s.\ there exist antipodal vertices $u,v$ of $T'_t$ with weights at least $\frac{a^*}{\rho^*}(1 - \eps) t$.
\end{lemma}

The proof is very similar to the proof of~\cite[Lemma~4]{treeheight} except a small twist is needed at the end to handle the negative weights.

\begin{proof}
Let $c=\frac{a^*}{\rho^*}$, and let $\eps,\delta>0$ be arbitrary. We prove that with probability at least $1-\delta$ 
{for all large enough $t$ there exists a pair $(u,v)$ of antipodal vertices of $T'_{\infty}$ with $\max\{B_u,B_v\} < t $ and $\min\{W_u,W_v\} > \left( 1 - 2\eps\right) ct$.}

Let $L$ be a constant positive integer that will be determined later,
and let $\alpha = \alpha^*$ and $\rho = \frac{\alpha}{c(1-\epsilon)} > \rho^*$.
By \eq{alpharho_lower} and since $\rho^*<1$ and $\Upsilon$ is strictly decreasing on $(0,1]$, we have
\begin{equation*}
\gamma_L(\alpha) + \Upsilon(\rho) < \log 2 \:.
\end{equation*}

Build a Galton-Watson process from $T'_{\infty}$ whose particles are a subset of vertices of $T'_{\infty}$, as follows.
Start with the root as the initial particle of the process.
If a given vertex $u$ is a particle of the process, then its potential offspring are its $2^L$ descendants that are $L$ levels deeper. Moreover, such a descendent $v$ is an offspring of $u$ if and only if $W_v - W_u \ge \alpha L$ and $B_v - B_u \le \rho L$.
As these two events are independent, the expected number of children of $u$ is at least
$$ 2^L \p{W_v - W_u \ge \alpha L} \p{B_v - B_u \le \rho L}
\ge \exp\left[ (\log 2 - \gamma_L (\alpha) - \Upsilon(\rho) - o(1)) L\right]
$$
as $L\to \infty$, by \eq{lower_bound_condition} and Lemma~\ref{lem_gamma}.
Since we have $\log 2 - \gamma_L (\alpha) - \Upsilon(\rho)>0$,
we may choose $L$ large enough that this expected value is strictly greater than 1.
Therefore, this Galton-Watson process survives with probability $q>0$.

We now boost this probability up to $1-\delta$, by starting several independent processes, giving more chance that at least one of them {survives}.
Specifically, let $b$ be a constant large enough that
$$(1-q)^{2^{b-1}} < \delta/3 \:.$$
Consider $2^{b}$ Galton-Watson processes, which have the vertices at depth $b$ of $T'_{\infty}$ as their initial particles, and reproduce using the same rule as before.
Let $a$ be a constant large enough that
$$ 2^{b+1} (e^{-a} + (1-p)^{a+2}) < \delta / 3 \:,$$
and let $A$ be the event that all edges $e$ in the top $b$ levels of $T'_{\infty}$ have $E_e \le a$ and $W_e \ge -a$. Then
$$1 - \p{A} \le 2^{b+1} (e^{-a} + (1-p)^{a+2}) < \delta / 3 \:.$$
Also, let $Q$ be the event that in each of the two branches of the root, at least one of the $2^{b-1}$ Galton-Watson processes survives.
Then
$$1 - \p{Q} \le 2 (1-q)^{2^{b-1}} < 2 \delta / 3\:,$$
and so with probability at least $1-\delta$ both $A$ and $Q$ occur.

Assume that both $A$ and $Q$ occur.
Let
$$m = \left \lfloor \frac{t (1-\eps) }{\rho L } \right \rfloor$$
and let $u$ and $v$ be particles at generation $m$ of surviving processes in distinct branches of the root.
Then $u$ and $v$ are antipodal,
$$\max\{B_u,B_v\} \le a b + m \rho L \le t (1-\eps) + O(1) < t \:,$$
and
$$\min\{W_u,W_v\} \ge -a b + m \alpha L \ge \frac{(1-\eps)\alpha}{\rho} \: t - O(1) > c (1-2\eps) t$$
for $t$ large enough, as required.
\end{proof}

Let $Y_1,Y_2, \dots$ be i.i.d.\ with $Y_i = 1 - \geo(p)$.
Recall the definition of $f : (-\infty,1] \to \mathbb{R}$ 
from \eq{f_def}:
\begin{equation*}
f(x) = (2-x)^{2-x} p (1-p)^{1-x} (1-x)^{x - 1} \:.
\end{equation*}
Note that $f(1)=p$ since {by convention} $0^0=1$, and  $f\left(2-p^{-1}\right)=1$.
The following lemma follows by noting that $f$ is positive and the derivative of $\log f$ is $\log \left ( \frac{1-x}{(2-x)(1-p)}\right)$.

\begin{lemma}
\label{lem:f_properties}
The function $f$ is continuous in $(-\infty,1]$ and differentiable in $(-\infty,1)$.
Moreover, $f$ is increasing on $(-\infty,2-p^{-1}]$
and decreasing on
$[2-p^{-1},1]$.
\end{lemma}


\begin{lemma}
\label{lem:ylargedev}
(a)
There is an absolute constant $C$ such that for any $a \in [2-p^{-1},1]$
and any positive integer $m$ we have
$$ \p{Y_1+\dots+Y_m \ge am} \le C {m} f(a)^m \:.$$

(b)
As $m\to\infty$, uniformly for all $a\in[0,1]$ we have
$$ \p{Y_1+\dots+Y_m \ge am} \ge \left[ f(a) - o(1) \right]^m \:.$$

(c)
If $p\ge 1/2$, then
as $m\to\infty$, uniformly for all $a\in[0,2-\frac{1}{p}]$ we have
$$ \p{Y_1+\dots+Y_m \ge am} \ge \left[ 1-o(1) \right] ^ m \:.$$
\end{lemma}

\begin{proof}
The conclusions are easy to see for $a=1$, so assume that $a<1$.
First, assume that $am$ is an integer.
Consider a sequence of independent biased coin flips, each of which is heads with probability $p$. A random walker   starts from 0,  takes one step to the  right on seeing heads, and one to the left on seeing   tails.
Then $Y_1+\dots+Y_m$ is the walker's position just after seeing the $m$-th head. Thus $Y_1+\dots+Y_m =am$ if and only if the $(2m-am)$-th coin comes up heads, and in the first $2m-am$ coin flips we see exactly $m$ heads and $m-am$ tails, so we have
\begin{align}
\p{Y_1+\dots+Y_m =am} & = \binom{2m-am-1}{m-1} p^m (1-p)^{m-am} \nonumber \\
& = \Theta \left ( \binom{2m-am}{m} p^m (1-p)^{m-am}  \right )  \nonumber \\
& = \Theta \left ( f(a)^m / \sqrt{m}\right ) \label{theta}\:,
\end{align}
where we have used Stirling's approximation for the last equality.

(a)
Let $a \in [2-{p}^{-1},1)$, and let $C$ be an absolute constant for the upper bound of $\Theta$ in \eq{theta}.
Then
\begin{align*}
\p{Y_1+\dots+Y_m \ge am} & \le m \sup \{\p{Y_1+\dots+Y_m = \alpha m} : \alpha \in[a,1]\}\\
& \le C \sqrt{m} \left[ \sup \{ f(\alpha) :  \alpha \in[a,1]\} \right] ^ m
\le C  m (f(a)) ^ m
\end{align*}
since  $f$ is decreasing on $\left[2-\frac{1}{p},1\right]$ by Lemma~\ref{lem:f_properties} and $m$ is a positive integer.

(b)
Assume that $m\to\infty$. Then
$$\p{Y_1+\dots+Y_m \ge am} \ge\p{Y_1+\dots+Y_m = \lceil am\rceil} = (f(a)-o(1))^m$$
uniformly for all $a\in[0,1)$ by continuity of $f$.

(c)
Assume that $p\ge 1/2$ and that $m\to\infty$.
Then
\begin{align*}
\p{Y_1+\dots+Y_m \ge am} & \ge\p{Y_1+\dots+Y_m = \left\lceil \left(2-{p^{-1}}\right) m \right\rceil}
\\
& = (f\left(2-{p^{-1}}\right)-o(1))^m = (1-o(1))^m
\end{align*}
uniformly for all $a\in[0,2-\frac{1}{p}]$
by continuity of $f$ and since $f(2-p^{-1})=1$.
\end{proof}

We define a two variable function
\begin{equation}
\label{Phi}
\Phi(a,s) =  p(1-p)(2-s)^2 (s-a)  - a(1-s) \:,
\end{equation}
and we define a function $\phi:[0,1]\to[0,1]$ as follows:
given $a\in[0,1]$, $\phi(a)$ is the unique solution in $[a,1]$ to
\begin{equation}
\label{phi}
\Phi(a,\phi(a)) =  p(1-p)(2-\phi(a))^2 (\phi(a)-a) - a(1-\phi(a)) = 0\:.
\end{equation}
Lemma~\ref{lem:s}(a) below shows that $\phi$ is well defined.
The proof of this lemma is straightforward and can be found in the appendix.

\begin{lemma}
\label{lem:s}
(a) Given $a\in[0,1]$, there is a unique solution $s\in[a,1]$ to $\Phi(a,s)=0$.
If $a\in\{0,1\}$ then $\phi(a) = a$.
If $a\in(0,1)$ then $0 < a < \phi(a) < 1$.

(b) If $s=\phi(a)$ then
$$\frac{sf(s)^{a/s}}{a^{a/s}(s-a)^{1-\frac{a}{s}}} = \frac{s}{s-a}\left(\frac{1-s}{(1-p)(2-s)}\right)^{a} \:.$$

(c) The function $\phi$ is increasing on $[0,1]$ and differentiable on $(0,1)$.

(d) The function $\phi$ is invertible and $\phi^{-1}$ is increasing.
If $s \in\{0,1\}$ then $\phi^{-1}(s)=s$.
If $s \in (0,1)$ then $0 < \phi^{-1}(s) < s < 1$.
\end{lemma}

Next let $\hat{Y}_1,\hat{Y}_2,\dots$ be independent and distributed as follows:
for every $i=1,2,\dots$ we flip an unbiased coin,
if it comes up heads, then $\hat{Y}_i=Y_i$,
otherwise $\hat{Y}_i=0$.

Define the function $g_L : (0,1) \to \mathbb{R}$ as
\begin{equation*}
g_L(a) =
\begin{cases}
1/2 & \mathrm{if\ } p > 1/2 \mathrm{\ and\ } 0 < a < 1 - \frac{1}{2p} \\
\frac{\phi(a)-a}{\phi(a)} \left( \frac{(1-p) (2-\phi(a))}{1-\phi(a)} \right)^a & \mathrm{otherwise.}
\end{cases}
\end{equation*}
Note that $g_L$ is continuous as $\phi(1 - \frac{1}{2p}) = 2 - 1/p$.
The proofs of the following two lemmas are standard and can be found in the appendix.

\begin{lemma}
\label{lem:large_deviation_lower}
We have the following large deviation inequality for every fixed $a\in(0,1)$ as $m\to \infty$.
\begin{equation*}
\p{\hat{Y}_1+\dots+\hat{Y}_m \ge am} \ge (2g_L(a)-o(1))^{-m}\:.
\end{equation*}
\end{lemma}

\begin{lemma}
\label{lem:uniques}
(a) There exists a unique solution $p_0 \in (0,1/2)$ to
$$\log \left( \frac{1-p}{p} \right) = \frac{1-p}{1-2p} \:.$$
Also,
if $p\le p_0$ then
$\log \left( \frac{1-p}{p} \right) \ge \frac{1-p}{1-2p}$.

(b)
Given $p\in(0,1)$, there exists a unique solution $s_0\in(0,1)$ to
\begin{equation*}
(1-p)(2-s) = \exp(1/s) (1-s)\:.
\end{equation*}
Moreover, if $p>1/2$ then $s_0 >2-p^{-1}$,
and if $p_0<p\le 1/2$ then $s_0 > \frac{1-2p}{1-p}$.
\end{lemma}

\begin{lemma}
\label{lem:main_lower}
Given $\eps>0$, a.a.s\ as $t\to\infty$ there exist two antipodal vertices $u,v$ of $T'_t$ with weights at least $c_L(p)(1 - \eps) t$. In particular, a.a.s.\ the weighted height of $T'_t$ is at least $c_L(p)(1 - \eps) t$.
\end{lemma}

\begin{proof}
By Lemma~\ref{lem:uniques}(b), there is a unique solution $s\in(0,1)$ to
\begin{equation*}
(1-p)(2-s) = \exp(1/s) (1-s)\:.
\end{equation*}
By the definition of $c_L$,
\begin{equation*}
c_L = c_L(p) = \exp(1/s) s (2-s) p \:.
\end{equation*}
Lemma~\ref{lem:large_deviation_lower} implies that the assumption~\eq{lower_bound_condition} of Lemma~\ref{lem_bd_lower} holds for $\gamma_L(a) = \log (2g_L(a))$. Let $a = \phi^{-1}(s)$ and let $\rho = 1 - \frac{a}{s}$. Since $s\in(0,1)$ we have $0 < a < s < 1$ by Lemma~\ref{lem:s}(d), and thus $\rho \in (0,1)$ as well.
Moreover, since $\Phi(a,s)=0$, we have $c_L = a / \rho$.

We now show that
$g_L(a) = \frac{s-a}{s} \exp (a/s)$.
This is clear if $p\le 1/2$, so assume that $p>1/2$.
It is easy to verify that $\Phi(1-\frac{1}{2p},2-\frac{1}{p}) = 0$. 
Since $p > 1/2$, by Lemma~\ref{lem:uniques}(b) we have $s > 2 - \frac{1}{p}$. Since $\phi^{-1}$ is increasing,
we have $a = \phi^{-1}(s) \ge 1- \frac{1}{2p}$. 

From $g_L(a) = \frac{s-a}{s} \exp (a/s)$ we get
$$\log (2 g_L(a)) + \rho -1 - \log (\rho) = \log 2 \:,$$
and Lemma~\ref{lem_bd_lower} completes the proof.
\end{proof}

The lower bound in Theorem~\ref{thm:main} follows easily from Lemmas~\ref{lem:main_lower} and~\ref{lem:equal_logs}.

\begin{proof}[Proof of the lower bound in Theorem~\ref{thm:diameter}.]
Fix $\eps>0$. Let us define the \emph{semi-diameter} of a tree as the maximum weighted distance between any two antipodal vertices. Clearly, semi-diameter is a lower bound for the diameter, so we just need to show a.a.s.\ as $n\to\infty$ the semi-diameter of the random-surfer model with $n$ vertices is at least $(2 c_L(p) - \eps) \log n$.
By Lemma~\ref{lem:main_lower}, a.a.s\ as $t\to\infty$ the semi-diameter of $T'_t$ is at least $(2 c_L(p) - \eps) t$. {Using} an argument similar to the proof of Lemma~\ref{lem:equal_logs} we may conclude that a.a.s.\ as $n\to\infty$ the semi-diameter of the third model (of Section~\ref{sec:transform}) with $2n-1$ vertices is at least $(2 c_L(p)-\eps)\log n$. 
It is easy to observe that this statement is also true for the random-surfer model with $n$ vertices, and the proof is complete.
\end{proof}

\section{Upper bounds for the random-surfer tree model}
\label{sec:upper}
In this section we prove the upper bounds in Theorems~\ref{thm:main} and~\ref{thm:diameter}.

\begin{lemma}
\label{lem:bd_upper}
Let $\gamma_U : [0,1] \to [0,\infty)$ be a continuous function such that for every fixed $a\in [0,1]$ and every vertex $v$ of $T_{\infty}$ at depth $m$,
\begin{equation}
\label{upper_bound_condition}
\p{\sum_{e\in \pi(v) } W_e > a m } \le \exp (-m \gamma_U (a) + o(m))
\end{equation}
as $m\to \infty$.
Define
\begin{equation}
\label{theta_def}
\theta = \sup \left \{ \frac{a}{\rho} : \gamma_U(a) + \Upsilon(\rho) = \log 2 :
a\in[0,1],\rho\in(0,\infty) \right \} \:.
\end{equation}
Then for every fixed $\eps>0$,
$$\p{\wh(T_t) > \theta(1 + \eps) t} \rightarrow 0$$
as $t\to \infty$.
\end{lemma}

The proof is similar to the proof of~\cite[Lemma~3]{treeheight}, in which the assumption \eq{upper_bound_condition} is not needed. In fact, in the model studied in~\cite{treeheight}, the weights $\{W_e : e\in \pi(v) \}$ are mutually independent, and the authors use Cram\'{e}r's Theorem to obtain a large deviation inequality for $\sum_{e\in \pi(v)} W_e$, which is similar to \eq{upper_bound_condition}.

\begin{proof}
We first prove a claim.
\begin{claim}
For every $\eps>0$ there exists $\delta>0$ such that for all $\rho \in \left(0,\frac{1}{\theta(1+\eps)}\right]$,
$$ \Upsilon(\rho) + \gamma_U (\theta(1+\eps)\rho) - \log 2 \ge \delta \:.$$
\end{claim}
\begin{claimproof}
Assume that this is not the case for some $\eps>0$.
This means there exists a sequence $(\rho_i)_{i=1}^{\infty}$ such that for all $i=1,2,\dots$,
$$ \Upsilon(\rho_i) + \gamma_U (\theta(1+\eps)\rho_i) - \log 2 < 1/i \:.$$
Then $(\rho_i)_{i=1}^{\infty}$ has a convergent subsequence. Let $\rho^* \in\left[0,\frac{1}{\theta(1+\eps)}\right]$ be the limit. It cannot be the case that
$\rho^* = 0$ since $\Upsilon(x) \to \infty$ as $x\to 0$, and $\gamma_U$ is non-negative. By continuity of $\Upsilon$ and $\gamma_U$ we have
$$ \Upsilon(\rho^*) + \gamma_U (\theta(1+\eps)\rho^*) - \log 2 \le 0 \:.$$
Since $\Upsilon$ is continuous, decreasing, and attains all values in $[0,\infty)$, we can choose $\rho' \le \rho^*$ so that
$$ \Upsilon(\rho') + \gamma_U (\theta(1+\eps)\rho^*) - \log 2 = 0 \:,$$
But then
$$\frac{\theta(1+\eps)\rho^*}{\rho'} \ge \theta(1+\eps) > \theta\:,$$
contradicting the definition of $\theta$ in \eq{theta_def}.
\end{claimproof}

Fix $\eps > 0$ and let $A_k$ be the event that there exists a vertex at depth $k$ of $T_t$ with weight larger than $\theta (1+ \eps)t$. By the union bound,
$$\p{\wh(T_t) > \theta(1 + \eps) t} \le \sum_{k=1}^{\infty} \p{A_k}
=\sum_{k > \theta (1+ \eps)t } \p{A_k}\:,$$
as the weights of all edges are at most 1.

Let $k > \theta (1+ \eps)t$. A vertex $v$ at depth $k$ of $T_{\infty}$  is included in $T_t$ and has weight larger than $\theta(1 + \eps)t$ if and only if $B_v \le t$ and $W_v > \theta(1 + \eps)t$. These two events are independent by the definition of $T_{\infty}$.
The random variable $B_v$ is distributed as a sum of $k$ independent exponential random variables with mean 1, and so
\begin{align*}
\p { B_v \le t, W_v > \theta (1+ \eps)t } & \le
\exp \left[ \left(-\Upsilon(t/k) - \gamma_U \left(\frac{\theta (1+ \eps)t}{k}\right) + o(1)\right)k \right ] \\
& \le \exp \left[(-\log 2 - \delta + o(1))k \right ] \:,
\end{align*}
where we have used Lemma~\ref{lem_gamma} and \eq{upper_bound_condition} for the first inequality, and $\delta>0$ is the constant provided by the claim.
{Since} there are $2^k$ vertices at depth $k$ of $T_\infty$, by the union bound
$$\p{A_k} \le 2^k \exp \left[ (-\log 2 - \delta + o(1))k \right ] \le
\exp \left[ (-\delta + o(1))k \right ] \:.$$
For $t$ large enough the $o(1)$ term is less than $\delta/2$, and thus
\begin{align*}
\p{\wh(T_t) > \theta(1 + \eps) t} & \le \sum_{k > \theta (1+ \eps)t} \p{A_k} \le \sum_{k> \theta (1+ \eps)t} \exp \left[ (-\delta/2) k \right ]
\\
& = O \left( \exp \left[ (-\delta/2) \theta (1+ \eps)t \right ] \right) =o(1)\:.\qedhere
\end{align*}
\end{proof}

Let $Y_1,Y_2, \dots$ be i.i.d.\ with $Y_i = 1 - \geo(p)$, and define random variables $X_1,X_2,\dots$ as follows:
$$X_1 = \max \{Y_1, 1\} \:,$$
and for $i\ge 1$,
$$X_{i+1} = \max \{ Y_{i+1}, 1 - (X_1+\dots+X_i)\} \:.$$
Define the function $h:[0,1]\to \mathbb{R}$ as
\begin{equation}
\label{h_def}
h(x) = \begin{cases}
1 & \mathrm{\ if\ } p \ge \frac{1}{2} \mathrm{\ and\ } 0 \le x \le 2 - \frac{1}{p} \\
\left(\frac{p}{1-p}\right)^x & \mathrm{\ if\ } p < \frac{1}{2} \mathrm{\ and\ } 0 \le x \le \frac{1-2p}{1-p} \\
(2-x)^{2-x} p (1-p)^{1-x} (1-x)^{x - 1} & \mathrm{otherwise} \:.
\end{cases}
\end{equation}
Note that in the third case we have $h(x)=f(x)$,
where $f$ is defined in \eq{f_def}.
It is easy to see that $h$ is continuous.
The proof of the following lemma can be found in the appendix.

\begin{lemma}
\label{lem:x_large_dev}
There exists an absolute constant $C$ such that
for every $a\in[0,1]$ and every positive integer $m$ we have
$\p{X_1+\dots+X_m > am} \le C m^2 h(a)^m$.
\end{lemma}

Next we define random variables $\hat{X}_1,\hat{X}_2,\dots$ as follows: for every $i=1,2,\dots$ we flip an independent unbiased coin, if it comes up heads, then $\hat{X}_i=X_i$,
otherwise $\hat{X}_i=0$.

We define the function $g_U : [0,1] \to \mathbb{R}$ as
\begin{equation}
\label{gu}
g_U(a) =
\begin{cases}
1/2 & \mathrm{if\ }p\ge 1/2 \mathrm{\ and\ } 0 \le a \le 1 - 1/2p \\
(\frac{1-p}{p})^a/2 & \mathrm{if\ } p < 1/2 \mathrm{\ and\ } 0 \le a \le \frac{1-2p}{2-2p} \\
\frac{1}{p} & \mathrm{if\ } a=1 \\
\frac{\phi(a)-a}{\phi(a)} \left( \frac{(1-p) (2-\phi(a))}{1-\phi(a)} \right)^a & \mathrm{otherwise} \:,
\end{cases}
\end{equation}
where $\phi$ is defined by \eq{phi}.
Note that by Lemma~\ref{lem:s}(a),
we have $ 0 < a < \phi(a) < 1$ for $a\in(0,1)$,
so $g_U$ is well defined for all $a\in[0,1]$.
The proof of the following lemma can be found in the appendix.

\begin{lemma}
\label{thm:large_deviation_upper}
(a) We have the following large deviation inequality for every $a\in[0,1]$ and every {positive integer} $m$, where $C'$ is an absolute constant:
\begin{equation}
\label{large_deviation_upper}
\p{\hat{X}_1+\dots+\hat{X}_m > am} \le
C' m^3 (2g_U(a))^{-m}\:.
\end{equation}

(b) The function $g_U$ is continuously differentiable on $(0,1)$ and
\begin{equation}
\label{gUder}
g'_U(a) =
\begin{cases}
0 &  \mathrm{if\ }p\ge 1/2 \mathrm{\ and\ } 0 < a \le 1 - 1/2p \\
\log (\frac{1-p}{p}) g_U(a) &  \mathrm{if\ }p < 1/2 \mathrm{\ and\ } 0 < a \le \frac{1-2p}{2-2p} \\
\log \left(\frac{(1-p)(2-\phi(a))}{1-\phi(a)}\right) g_U(a) & \mathrm{otherwise} \:.
\end{cases}
\end{equation}

(c)
The function $\log g_U(a)$ is increasing and convex. It is strictly increasing when $g_U(a)>1/2$.
\end{lemma}

\begin{lemma}
\label{lem:nu}
Let $\omega > 0$ and let $\tau(x) : [0,\omega] \rightarrow \mathbb{R}$ be a positive function
that is differentiable on $(0,\omega)$ and satisfies
\begin{equation}
\label{tau}
\alpha(x) + \chi(\tau(x)) = 0 \qquad \forall\:x\in[0,\omega]
\end{equation}
for convex functions $\alpha,\chi$,
with $\alpha$ increasing and $\chi$ decreasing.
Assume there exists $x^*\in(0,\omega)$ such that
$\tau'(x^*) = \tau(x^*) / x^*$.
Then we have
\begin{equation}
\label{x*y}
\frac{x^*}{\tau(x^*)} \ge \frac{y}{\tau(y)}
\end{equation}
for all $y \in [0,\omega]$.
\end{lemma}

\begin{proof}
We first prove that $\tau$ is convex and increasing.
Pick $x_1,x_2 \in [0,\omega]$ and $\lambda_1, \lambda_2 \in [0,1]$ with $\lambda_1+\lambda_2=1$.
We need to show that
\begin{equation}
\label{ntstau}
\tau(\lambda_1 x_1 + \lambda_2 x_2) \le \lambda_1 \tau(x_1) + \lambda_2 \tau(x_2) \:.
\end{equation}
We have
\begin{align*}
\chi ( \lambda_1 \tau(x_1) + \lambda_2 \tau(x_2) ) & \le
\lambda_1 \chi(\tau(x_1)) + \lambda_2 \chi(\tau(x_2)) \\
& = - \lambda_1 \alpha(x_1) - \lambda_2 \alpha(x_2) \\
& \le - \alpha (\lambda_1 x_1 + \lambda_2 x_2) \\
& = \chi(\tau(\lambda_1 x_1 + \lambda_2 x_2))
\end{align*}
by convexity of $\chi$, then \eq{tau}, then convexity of $\alpha$, and then \eq{tau} again. The equation \eq{ntstau} follows since $\chi$ is decreasing. Hence $\tau$ is convex. Also, $\tau$ is increasing since $\alpha$ is increasing and $\chi$ is decreasing.

Now, let $y \in [0,\omega]$.
We prove \eq{x*y} for $y < x^*$. The proof for $y>x^*$ is similar.
By the mean value theorem, there exists $z\in (y, x^*)$ with
\begin{equation*}
\tau'(z) = \frac{\tau(x^*) - \tau(y)}{x^* - y} \:.
\end{equation*}
On the other hand, since $z < x^*$ and $\tau$ is convex, we have
$$
\tau'(z) \le \tau'(x^*) = \frac{\tau(x^*)}{x^*} \:. $$
The inequality \eq{x*y} follows from these two results.
\end{proof}

We are ready to prove the upper bound in Theorem~\ref{thm:main}. The upper bound in Theorem~\ref{thm:diameter} follows {immediately} as in every tree the diameter is at most twice the height.

\begin{proof}[Proof of the upper bound in Theorem~\ref{thm:main}.]
Let $c_U = c_U(p)$.
By Lemma~\ref{lem:equal_logs} we just need to show that given $\eps>0$, a.a.s\ as $t\to\infty$ the weighted height of $T_t$ is at most $ (1+ \eps) c_U t$. For proving this we use Lemma~\ref{lem:bd_upper}.  Lemma~\ref{thm:large_deviation_upper} implies that condition \eq{upper_bound_condition} of Lemma~\ref{lem:bd_upper} holds with $\gamma_U(a) = \log (2g_U(a))$, so we need only show that
\begin{equation}
\label{ntscu}
c_U = \sup \left \{ \frac{a}{\rho} :
\log(g_U(a)) + \rho -1 - \log (\rho) = 0:
a\in[0,1],\rho\in(0,\infty) \right \}.
\end{equation}
The function $\rho -1 - \log (\rho)$ attains all values in $[0,\infty)$ for $\rho\in(0,1]$.
Moreover, it is strictly decreasing for $\rho\in(0,1]$ and equals 0 for $\rho \in [1,\infty)$.
So $\log(g_U(a)) + \rho -1 - \log (\rho) = 0$ has a unique solution (for $\rho$) if
$0<g_U(a)< 1$,
and no solution if $g_U(a) > 1$.
Since 
$g_U(0)=1/2$ and $g_U(1) = 1/p$, and the function
$g_U(x)$ is continuous and strictly increasing when $g_U(x)>1/2$, there is a unique $x$ with $g_U(x) = 1$.
Denote this point by $a_{\max}$.
Define the function $\tau : [0,a_{\max}] \to (0,1]$ as follows. Let $\tau(a_{\max}) = 1$ and for $x<a_{\max}$ let $\tau(x)$ be the unique number satisfying
\begin{equation}
\label{def_nu}
\log(g_U(x)) + \tau(x) -1 - \log \tau(x) = 0 \:.
\end{equation}
Hence to prove \eq{ntscu} it is enough to show that
\begin{equation}
\label{ntsnu}
c_U = \sup \left \{ \frac{x}{\tau(x)} : x\in[0,a_{\max}] \right \} \:.
\end{equation}

We prove \eq{ntsnu} using Lemma~\ref{lem:nu}.
The function $\log(g_U(a))$ is increasing and convex by Lemma~\ref{thm:large_deviation_upper}(c), and it is easy to check that the function $\rho -1 - \log (\rho)$ is decreasing and convex.
Moreover, differentiating \eq{def_nu} gives
$$
\frac{g'_U(x)}{g_U(x)} + \tau'(x) - \frac{\tau'(x)}{\tau(x)} = 0 \:.
$$
So by the implicit function theorem $\tau$ is differentiable in $x\in(0,a_{\max})$ and
$$\tau'(x) = \frac{\tau(x)}{1-\tau(x)} \: \frac{g'_U(x)}{g_U(x)} \:.$$
By Lemma~\ref{lem:nu}, we just need to show the existence of
$x^* \in (0, a_{\max})$ with
\begin{equation}
\label{ntslast}
c_U = \frac{x^*}{\tau(x^*)} = \frac{1-\tau(x^*)}{\tau(x^*)} \: \frac{g_U(x^*)}{g'_U(x^*)}\:.
\end{equation}

We consider two cases.
Recall that $p_0 \approx 0.206$ is the solution to
\begin{equation*}
\log \left( \frac{1-p}{p} \right) = \frac{1-p}{1-2p} \:,
\end{equation*}
which has a unique solution by Lemma~\ref{lem:uniques}(a).

\noindent\textbf{Case 1: $0 < p \le p_0$.}
In this case we have
$$c_U = \left( \log \left( \frac{1-p}{p} \right) \right)^{-1} \:.$$
Let
$$a^* = \left [ 2 \log \left( \frac{1-p}{p} \right) \right ]^{-1} \:.$$
By Lemma~\ref{lem:uniques}(a)
$$\log \left (\frac{1-p}{p}\right) \ge \frac{1-p}{1-2p} \:,$$
which gives $a^* \le \frac{1-2p}{2-2p}$, thus
$$g_U(a^*) = \left(\frac{1-p}{p}\right)^{a^*}\Big/2 = \exp\left(\frac{1}{2} - \log 2\right)<1$$
by the definition of $g_U$ in \eq{gu}, and
$$g'_U(a^*) = \log \left ( \frac{1-p}{p} \right) g_U(a^*)$$
by Lemma~\ref{thm:large_deviation_upper}(b).
The definition of $\tau$ in \eq{def_nu} implies $\tau(a^*) = 1/2$.
Moreover,
$$ \frac{a^*}{\tau(a^*)}
= \left( \log \left( \frac{1-p}{p} \right) \right)^{-1}
= \frac{1-\tau(a^*)}{\tau(a^*)} \: \frac{g_U(a^*)}{g'_U(a^*)}
\:,$$
which gives \eq{ntslast}.
Finally, since $g_U(a^*) < 1$, we have $a^* \in (0,a_{\max})$, and the proof is complete.

\noindent\textbf{Case 2: $p_0 < p < 1$.}
In this case we have
$$ c_U = p s^* (2-s^*) \exp (1/s^*) \:,$$
where $s^*\in(0,1)$ is the unique solution for
\begin{equation}
\label{ss}
s^* \log \left ( \frac{(1-p)(2-s^*)}{1-s^*} \right) = 1 \:.
\end{equation}
Lemma~\ref{lem:uniques}(b) implies that $s^*$ is well defined.
Let $a^* = \phi^{-1}(s^*)$.

We first show that
\begin{equation}
\label{correctrange}
g_U(a^*) = \frac{s^*-a^*}{s^*} \left( \frac{(1-p) (2-s^*)}{1-s^*} \right)^{a^*} \:.
\end{equation}

If $p > 1/2$, then by Lemma~\ref{lem:uniques}(b) we have $s^* > 2 - \frac{1}{p}$.
It is easy to verify that $\Phi(1-\frac{1}{2p},2-\frac{1}{p}) = 0$.
Since $\phi^{-1}$ is increasing, we have $a^* = \phi^{-1}(s^*) > 1- \frac{1}{2p}$,
so \eq{correctrange} agrees with the definition of $g_U$ in \eq{gu}.

If $p_0 < p \le 1/2$, then by Lemma~\ref{lem:uniques}(b) we have $s^* > \frac{1-2p}{1-p}$.
It is easy to verify that $\Phi(\frac{1-2p}{2-2p},\frac{1-2p}{1-p}) = 0$.
Since $\phi^{-1}$ is increasing, we have $a^* = \phi^{-1}(s^*) > \frac{1-2p}{2-2p}$,
so \eq{correctrange} agrees with the definition of $g_U$ in \eq{gu}.

Using \eq{ss}, the equation \eq{correctrange} simplifies into
\begin{equation}
\label{newone}
g_U(a^*) = \left( 1 - \frac{a^* }{s^* } \right) \exp \left (a^*/s^*\right) <
\exp ( - \frac{a^* }{s^* }) \exp \left (a^*/s^*\right) = 1\:,
\end{equation}
and by Lemma~\ref{thm:large_deviation_upper}(b) we have
$$g'_U(a^*) = g_U(a^*)\log \left(\frac{(1-p)(2-s^*)}{1-s^*}\right) = g_U(a^*)/s^* \:.$$
It follows from \eq{newone} and the definition of $\tau$ in \eq{def_nu} that $\tau(a^*) = 1 - \frac{a^*}{s^*}$.
Using \eq{ss} and  $\Phi(a^*,s^*)=0$, we get
$$ \frac{a^*}{\tau(a^*)}
= p s^* (2-s^*) \exp (1/s^*)
=\frac{1-\tau(a^*)}{\tau(a^*)} \: \frac{g_U(a^*)}{g'_U(a^*)}\:,$$
which gives \eq{ntslast}.
Finally, since $g_U(a^*) < 1$, we have $a^* \in (0,a_{\max})$, and the proof is complete.
\end{proof}

\section{Concluding Remarks}
\label{sec:conclude}

There is a common generalization of random recursive trees, preferential attachment trees, and random-surfer trees. Consider i.i.d.\ random variables $X_1,X_2,\ldots \in \{0,1,2,\dots\}$. Start with a single vertex $v_0$. At each step $s$ a new vertex $v_s$ appears, chooses a random vertex $u$ in the present graph, and then walks $X_s$ steps from $u$ towards $v_0$, joining to the last vertex in the walk (if it reaches $v_0$ before $X_s$ steps, it joins to $v_0$).
Random recursive trees correspond to $X_i = 0$,
preferential attachment trees correspond to $X_i=\operatorname{Bernoulli}(1/2)$ (see, e.g., \cite[Theorem~3.1]{random_surfer}), and
random-surfer trees correspond to $X_i=\geo(p)$.
Using the ideas of this paper, it is possible to obtain lower and upper bounds for the height and the diameter of this general model (similar to Theorems~\ref{thm:main} and~\ref{thm:diameter}), provided one can prove large deviation inequalities (similar to Lemma~\ref{lem:ylargedev}) for the sum of $X_i$'s  and also large deviation inequalities
(similar to Lemma~\ref{lem:x_large_dev}) for the sum of random variables $X'_i$, defined as
$$X'_1 = 1, \qquad X'_{i+1} = \max \{1 - X_i, 1 - (X'_1 + \dots + X'_i)\} \:.$$

\bibliographystyle{plain}
\bibliography{webgraph}	
\section*{Appendix: omitted proofs}
\label{sec:omit}

\begin{proof}
[Proof of Lemma~\ref{lem_gamma}.]
We first prove the upper bound.
If $x>1$ then $\exp (-\Upsilon(x) m)=1$, so we may assume that $0 < x \le 1$.
We use Chernoff's technique.
Let $\theta = 1 - 1/x$. Then we have
\begin{align*}
\p{E_1 + E_2 + \dots + E_m \le xm} &
=
\p{\exp(\theta E_1  + \dots + \theta E_m) \ge \exp(\theta xm)} \\
& \le
\e{\exp(\theta E_1  + \dots + \theta E_m)} / \exp(\theta x m)
\\
& 
= 
\e{\exp(\theta E_1)}  
\e{\exp(\theta E_2)}  
\dots
\e{\exp(\theta E_m)}  / \exp(\theta x m) \\
& = 
(1-\theta)^{-m} \exp(-\theta x m) = 
\exp (-\Upsilon(x) m) \:.
\end{align*}
We now prove the lower bound.
If $x>1$, then the result follows from Markov's inequality, so we may assume that $0<x\le 1$.
Let $\Lambda^*(x) = \sup \{ \lambda x - \log (\e{e^{\lambda E_1}}) : \lambda \le 0\} $.
Since
$\e{e^{\lambda E_1}}=1/(1-\lambda)$ for all $\lambda<1$,
the supremum here occurs at $\lambda = 1 - 1/x$,
which implies $\Lambda^*(x)=\Upsilon(x)$.
Then by Cram\'{e}r's Theorem (see, e.g., \cite[Theorem~2.2.3, p.~27]{dembo})
we have
$$\p{E_1 + E_2 + \dots + E_m \le xm} = \exp (-\Lambda^*(x) m + o(m)) = \exp (-\Upsilon(x) m + o(m))\:,$$
as required.
\end{proof}

\begin{proof}[Proof of Lemma~\ref{pro:geo}.]
We use Chernoff's technique.
Let $\theta$ satisfy 
$$e^{\theta} = \frac{\kappa-1}{\kappa(1-p)} \:.$$
We have
$$
 \e{\exp(\theta Z_1)}
 =\sum_{k=1}^{\infty}p(1-p)^{k-1}e^{\theta k}
 = \frac{pe^{\theta}}{1-e^{\theta}(1-p)} \:.
$$
Thus we have
\begin{align*}
\p{Z_1 + Z_2 + \dots + Z_m \ge \kappa m} &
=
\p{\exp(\theta Z_1  + \dots + \theta Z_m) \ge \exp(\theta \kappa m)} \\
& \le
\e{\exp(\theta Z_1  + \dots + \theta Z_m)} / \exp(\theta \kappa m)
\\
& 
= 
\e{\exp(\theta Z_1)}  
\e{\exp(\theta Z_2)}  
\cdots
\e{\exp(\theta Z_m)}  / \exp(\theta \kappa m) \\
& = 
\left(\frac{pe^{\theta-\theta\kappa}}{1-e^{\theta}(1-p)}\right)^m
= f(2-\kappa)^m \:.\qedhere
\end{align*}
\end{proof}

\begin{proof}[Proof of Lemma~\ref{lem:technical}.] We consider two cases.\\ 
\noindent\textbf{Case 1: $c\ge 1$.}
In this case we prove
$$-c \Upsilon(1/c) + c \log f ( 2 - \eta / c) < \eta(1-p)\log(1-p^3)-1 .$$
Notice that we have
$1-c\Upsilon(1/c) = c + c \log(1/c)$,
so, using the definition of $f$ and since $\eta(1-p)\le\eta-c$, the conclusion is implied by
$$
c + c \log(1/c) + \eta \log (\eta(1-p)/(\eta-c)) + c \log (p(\eta-c)/((1-p)c)) < (\eta-c)\log(1-p^3) \:.
$$
Letting $r = \eta / c$ and since $c>0$, this statement is equivalent to
$$
ep (1-p)^{r-1} r^2 (r/(r-1))^{r-1} < \eta (1-p^3)^{r-1}\:.
$$
Since $(r/(r-1))^{r-1} < e$, and $1-p < (1-p^3)e^{-p}$, for this inequality to hold it suffices to have
\begin{equation*}
\label{stp}
e^{2+p} r^2 p \exp(-pr) \le 4e^p/p \qquad \forall r\in[p^{-1},\infty) \:,
\end{equation*}
which follows from the fact that 
$x^2 e^{-x} \le 4e^{-2}$ for all $x\ge 1$.

\noindent\textbf{Case 2: $c< 1$.}
In this case we prove
$$-c \Upsilon(1/c) + c \log f ( 2 - \eta / c) < -0.15p\eta-1 .$$
Since $\Upsilon(1/c)=0$, this is equivalent to
\begin{equation}
1+0.15p \eta  + c \log f ( 2 - \eta / c) < 0 .
\label{case2nts}
\end{equation}
Note that 
$$
\left(\frac{\eta/c}{\eta/c - 1}\right)^{\eta/c - 1} < e\:,
$$
so we have
\begin{align*}
c \log f (2- \eta / c)
&
= \log \left((\eta/c)^{\eta}p^c(1-p)^{\eta-c}(\eta/c-1)^{c-\eta} \right)\\
& <
\log\left((e \eta p  / c)^c (1-p)^{\eta - c}  \right)
\le c \log (e \eta p / c) + c p - p \eta \:,
\end{align*}
where we have used $\log(1-p)\le -p$ in the last inequality.
Hence to prove \eq{case2nts}, since $c>0$, it suffices to show that
\begin{equation}
\label{lastlast}
\frac{1}{c}  + 1 + \log (\eta p / c) + p
< 0.85 p \eta / c \:.
\end{equation}
Since $ p\eta \ge 4e^p\ge 4>4c$, we have
\begin{align*}
\frac1c & < 0.25 p \eta / c,\\
1+p < \frac{1+p}{c} < \frac{e^p}{c} & \le 0.25 p\eta/c ,\\
\log (\eta p / c) & < 0.35 p \eta / c \:,
\end{align*}
which imply~\eq{lastlast}.
\end{proof}

\begin{proof}[Proof of Lemma~\ref{lem:s}.]
(a)
The conclusion is clear for $a\in\{0,1\}$, so we may assume that $a\in(0,1)$.
Since $\Phi(a,a)<0$ and $\Phi(a,1)>0$, there exists at least one $s\in(a,1)$ with
$\Phi(a,s)=0$. We now show that there is a unique  such $s$.
Fixing $a$, since $\Phi$ is differentiable with respect to $s$, it is enough to show that
\begin{equation}
\label{ntsphi}
\mathrm{if\ } \Phi(a,s)=0\mathrm{,\ then\ } \frac{\partial \Phi}{\partial s} > 0
\end{equation}
Let $\sigma = p(1-p) $. We have
$$
\frac{\partial \Phi}{\partial s} = \sigma (2-s)^2 + a - 2 \sigma (2-s)(s-a) \:.
$$
At a point $(a,s)$ with $\Phi(a,s)=0$, we have
$$ \sigma (2-s)^2 = \frac{a(1-s)}{s-a}\mathrm{,\ and\ } \sigma (2-s)(s-a) = \frac{a(1-s)}{2-s} \:,$$
so at this point,
\begin{equation*}
\frac{\partial \Phi}{\partial s} =
a \left( \frac{1}{s-a} + \frac{1}{1-s} - \frac{2}{2-s} \right) \:,
\end{equation*}
which is strictly positive because
$$\min \left\{ \frac{1}{s-a} , \frac{1}{1-s} \right \} > 1 > \frac{1}{2-s} \:,$$
and this proves \eq{ntsphi}.

(b) Plugging the definition of $f$ from \eq{f_def} and using $\Phi(a,s)=0$ gives this equation.

(c)
We first show that $\phi$ is differentiable and increasing on $(0,1)$.
Let $a\in(0,1)$ and let $s=\phi(a)$.
We have
$$\frac{\partial \Phi}{\partial a} = s-1-p(1-p)(2-s)^2 < 0\:,$$
and $\partial \Phi / \partial s$ is positive as proved in part (a).
Hence by the implicit function theorem
${\mathrm{d}s}/{\mathrm{d}a}$ exists and is positive,
so $\phi$ is differentiable and increasing on $(0,1)$.
Since $\phi(0)=0$ and $\phi(1)=1$, $\phi$ is increasing on $[0,1]$.

(d)
Let $s\in[0,1]$.
Then $\Phi(0,s) \Phi (s,s) \le 0$ and so there exists at least one $a_0\in[0,s]$
with $\Phi(a_0,s)=0$. The function $\Phi(a,s)$ is linear in $a$ and the coefficient of $a$ is non-zero, hence this  root $a_0$  is unique.
The function $\phi^{-1}$ is increasing since $\phi$ is increasing.
The last two statements follow from similar statements proved for $\phi$ in (a).
\end{proof}

\begin{proof}[Proof of Lemma~\ref{lem:large_deviation_lower}.]
 We have
$$
\p{\hat{Y}_1+\dots+\hat{Y}_m \ge am} = \sum_{k=\left\lceil am \right\rceil}^m \binom{m}{k} 2^{-m} \times \p{Y_1+\dots+Y_k\ge am} \:,
$$
where $k$ denotes the number of $\hat{Y}_i$'s whose value was determined to be equal to $Y_i$.

If $p> 1/2$ and $0 < a < 1 - \frac{1}{2p}$, then letting $k = \lceil m/2 \rceil$ gives
$$\binom{m}{k} 2^{-m}  = \Omega \left( \frac{1}{\sqrt m} \right)$$
by Stirling's approximation, and
$$ \p{Y_1+\dots+Y_k\ge am} \ge (1-o(1))^k$$
by Lemma~\ref{lem:ylargedev}(c).
This gives
$$
\p{\hat{Y}_1+\dots+\hat{Y}_m \ge am} \ge (1-o(1))^m\:,
$$
as required.

Otherwise, let $s=\phi(a)$.
Then letting $k = \lceil am/s \rceil$ gives
$$\binom{m}{k} 2^{-m}  = \Omega \left( \left[\frac{s(s-a)^{a/s}}{2(s-a)a^{a/s}} \right]^m\right) \Big/ m^2$$
by Stirling's approximation, and
$$ \p{Y_1+\dots+Y_k\ge am} \ge (f(s)-o(1)) ^ k$$
by Lemma~\ref{lem:ylargedev}(b) and since $f$ is continuous.
Lemma~\ref{lem:s}(b) completes the proof.
\end{proof}

\begin{proof}[Proof of Lemma~\ref{lem:uniques}.]
(a)
The function $r(p)=\log \left( \frac{1-p}{p} \right) - \frac{1-p}{1-2p}$
approaches $+\infty$ when $p\to 0^+$ and approaches $-\infty$ when $p\to { \frac{1}{2} }^{-}$.
Moreover,
$$r'(p)=\frac{-1}{p(1-p)}-\frac{1}{(1-2p)^2} < 0$$
for $p\in(0,1/2)$.
Hence $r(p)$ has a unique root $p_0$, and
$r(p) \ge 0$
if and only if $p\le p_0$.

(b)
The function
$$\mu(s) = \log (1-p) + \log(2-s) - \log(1-s) - \frac{1}{s}$$
approaches $-\infty$ as $s \to 0^+$,
and approaches $+\infty$ as $s \to 1^{-}$,
and its derivative is positive in $(0,1)$,
hence it has a unique root {$s_0$} in $(0,1)$.
Also we have $\mu(2-p^{-1}) = p/(1-2p)$,
which means that if $p>1/2$ then {$s_0>2-p^{-1}$}.
Moreover, if $p_0 < p \le 1/2$, then by part (a),
$$\mu\left(\frac{1-2p}{1-p}\right) = \log\left(\frac{1-p}{p}\right) - \frac{1-p}{1-2p}=r(p)<0\:,$$
which means $s_0 > \frac{1-2p}{1-p}$.
\end{proof}

\begin{proof}[Proof of Lemma~\ref{lem:x_large_dev}.]
The conclusion is obvious if $p\ge \frac{1}{2}$ and $a \le 2 - \frac{1}{p}$, or if $a=0$, since in these cases $h(a)=1$.
Also, $\p{X_1+\dots+X_m > m}=0$ so the conclusion is true if $a=1$,
so we may assume that $\max \{ 0, 2 - \frac{1}{p} \} < a < 1$.

Observe that if $X_1+\dots+X_m > a m$, there is a subsequence of the form
$Y_{m-k+1},\dots,Y_m$ whose sum is at least $am$, and this subsequence contains at least $am$ elements since $Y_i \le 1$ for all $i$. Hence we have
\begin{equation*}
\p{X_1+\dots+X_m > am} \le m \max\{ \p{Y_1+\dots+Y_k \ge am}: k \in [am,m] \cap \mathbb{Z}\}
\end{equation*}
as the $Y_i$'s are i.i.d.

For any integer $k \in [am, m]$, by Lemma~\ref{lem:ylargedev}(a) we have
$$ \p{Y_1+\dots+Y_k \ge am} \le C k (f(am/k))^k$$
for an absolute constant $C$, since $am/k \ge a > 2 - \frac{1}{p}$.
Let $r = k/m \in [a,1]$. So we find that
$$\p{X_1+\dots+X_m > am} \le C m^2 \big(\sup \{ f(a/r)^r  : r\in[a,1]\} \big)^m \:.$$
Let us define
$$ \xi(r) = f(a/r)^r = (2r-a)^{2r-a}p^r(1-p)^{r-a}(r-a)^{a-r}r^{-r} \:.$$
So to {complete the proof} we just need to show that
\begin{equation}
\label{xxdev}
\sup \{ \xi(r)  : r\in[a,1] \} \le h(a) \qquad \forall \: a \in \left(\max\left\{0,2 - \frac{1}{p}\right\},1\right) \:.
\end{equation}

The function $\xi(r)$ is positive and differentiable for each $a\in(0,1)$, hence the supremum here occurs either at a boundary point or at a point with zero derivative.
The derivative of $\log (\xi(r))$ equals
$$\log \left ( \frac{p(1-p)(2r-a)^2}{r(r-a)} \right ) \:.$$
Thus $\xi'(r)$ has the same sign as
$\overline{\xi}(r) = p(1-p)(2r-a)^2 - r(r-a)$ in $r\in[a,1]$.
Notice that $\overline{\xi}(r)$ has two roots
$$r_1 = \frac{ap}{2p-1} \mathrm{,\ and\ } r_2 = \frac{a(1-p)}{1-2p} \:.$$
We may consider several cases.

\noindent\textbf{Case 0: $p=1/2$.} The function $\overline{\xi}$ is positive, so $\xi$ is increasing in $[a,1]$, hence the supremum in~\eq{xxdev} happens at $r=1$ and its value is $f(a)$.

\noindent\textbf{Case 1: $p>1/2$.} Since $a > 2 - \frac{1}{p}$,
we find that $r_1>1$ and $r_2<0$. Moreover, $\overline{\xi}(a)\ge 0$.
Thus $\overline{\xi}$ is non-negative in $[a,1]$,
which implies $\xi$ is increasing in $[a,1]$.
Thus the supremum in~\eq{xxdev} happens at $r=1$ and its value is $f(a)$.

\noindent\textbf{Case 2: $p<1/2$ and $a \le \frac{1-2p}{1-p}$.}
In this case $r_1<0$ and $a \le r_2 \le 1$.
Since $\overline{\xi}(a) \ge 0$ and
$\overline{\xi}(r_1)=\overline{\xi}(r_2)=0$ and $\overline{\xi}$ is quadratic,
the function $\overline{\xi}$ goes from positive to negative at $r_2$.
Therefore, the function $\xi$ attains its supremum at $r_2$ and the supremum value in~\eq{xxdev} equals
$$\xi(r_2) = \left(\frac{p}{1-p}\right)^a \:.$$

\noindent\textbf{Case 3: $p<1/2$ and $a > \frac{1-2p}{1-p}$.}
We find that $r_1<0$ and $r_2>1$, and $\overline{\xi}(a) \ge 0$,
so $\overline{\xi}$ is non-negative in $[a,1]$, hence
$\xi$ is increasing in $[a,1]$.
Thus the supremum in~\eq{xxdev} happens at $r=1$ and its value is $f(a)$. This completes the proof of \eq{xxdev} and the lemma.
\end{proof}

\begin{proof}[Proof of Lemma~\ref{thm:large_deviation_upper}.]
(a)
First, the case $a=0$ is obvious since $g_U(0) = 1/2$, and
the case $a=1$ is easy since
$\p{\hat{X}_1+\dots+\hat{X}_m > m} = 0$. So we may assume that $a\in(0,1)$.

Letting $k$ of the $\hat{X}_i$'s being equal to $X_i$ and the rest equal to zero, we get
\begin{align*}
\p{\hat{X}_1+\dots+\hat{X}_m > am} & = \sum_{k=am}^m \binom{m}{k} 2^{-m} \p{X_1+\dots+X_k > am} \\
& \le m \sup \left\{ \binom{m}{rm} 2^{-m} \p{X_1+\dots+X_{rm} > am} :r\in[a,1] \right\} \:.
\end{align*}
For a given $r \in [a,1]$, Lemma~\ref{lem:x_large_dev} gives
$$ \p{X_1+\dots+X_{rm}> am} \le C(rm)^2 h(a/r)^{rm} \le C m^2 h(a/r)^{rm} \:.$$
Moreover, by Stirling's approximation
$$ \binom{m}{rm} = O\left( \frac{1}{r^{rm} (1-r)^{(1-r)m}} \right) \:.$$
So, we find {that}
$$\p{\hat{X}_1+\dots+\hat{X}_m > am} \le
 C' m^3 \left [ \sup \left\{ \frac{h(a/r)^r}{2 r^{r} (1-r)^{1-r}} :r\in[a,1] \right\} \right] ^m \:.$$
Thus to {complete the proof of part (a)} we just need to show
\begin{equation}
\label{ntsgU}
g_U(a) = \inf \left\{ \frac{\zeta-a}{\zeta} \left(\frac{a}{(\zeta-a)h(\zeta)}\right)^{a/\zeta}  : \zeta\in[a,1] \right\} \:,
\end{equation}
where we have used the change of variable $\zeta = a/r$.
For analysing this infimum we define the two variable function
\begin{equation*}
\psi (a,\zeta) = \frac{\zeta-a}{\zeta} \left(\frac{a}{(\zeta-a)h(\zeta)}\right)^{a/\zeta} \end{equation*}
with domain $\{(a,\zeta) : 0 < a < 1, a \le \zeta \le 1 \}$,
and consider two cases depending on the value of $p$.

\noindent\textbf{Case 1: $p \ge 1/2$.}
By the definition of $h$ in \eq{h_def} we have
$$
\psi(a,\zeta) =
\begin{cases}
\frac{\zeta-a}{\zeta} \left(\frac{a}{\zeta-a}\right)^{a/\zeta}  & \mathrm{\ if\ }a \le \zeta \le 2 - p^{-1} \\
\frac{\zeta-a}{\zeta} \left(\frac{a}{(\zeta-a)f(\zeta)}\right)^{a/\zeta}
& \mathrm{\ otherwise}\:,
\end{cases}
$$
where $f$ is defined in \eq{f_def}.
Since $f(2-p^{-1}) = 1$, $\psi$ is continuous here.
Let us define $\psi_1(\zeta) = \frac{\zeta-a}{\zeta} \left(\frac{a}{\zeta-a}\right)^{a/\zeta}$
and
$\psi_2(\zeta)=
\frac{\zeta-a}{\zeta} \left(\frac{a}{(\zeta-a)f(\zeta)}\right)^{a/\zeta}
$.

The derivative of $\log \psi_1(\zeta)$ is
$$ a \log \left( \frac{\zeta-a}{a} \right) / \zeta^2 \:,$$
which is negative for $\zeta<2a$ and positive for $\zeta>2a$.
This implies $\psi_1(\zeta)$ is decreasing when $\zeta \le 2a$ and
increasing when $\zeta \ge 2a$.
So $\psi_1$ achieves its minimum at $\zeta = 2a$, and its minimum value is $1/2$.

The derivative of $\log \psi_2(\zeta)$ is
$$ \frac{a}{\zeta^2} \left[\log \left ( p(1-p)(2-\zeta)^2 (\zeta-a) \right)  - \log \left( a(1-\zeta) \right) \right ]\:.
$$
Comparing with \eq{Phi} we find that this derivative has the same sign as $\Phi(a,\zeta)$.
So by Lemma~\ref{lem:s}(a) it vanishes at a unique point $\zeta = \phi(a)$.
Also at $\zeta=\phi(a)$ we have $\partial \Phi / \partial \zeta > 0$ (see \eq{ntsphi}), which implies $\Phi(a,\zeta)$ is non-positive
when $\zeta \le \phi(a)$
and non-negative when $\zeta \ge \phi(a)$.
Thus $\psi_2$ achieves its minimum at $\phi(a)$, and
its minimum value is
$$\psi_2(\phi(a)) = \frac{\phi(a)-a}{\phi(a)} \left(\frac{a}{(\phi(a)-a)f(\phi(a))}\right)^{a/\phi(a)}
 = \frac{\phi(a)-a}{\phi(a)} \left( \frac{(1-p) (2-\phi(a))}{1-\phi(a)} \right)^a$$
by Lemma~\ref{lem:s}(b).

We conclude that:

(i) If $2a \le 2 - 1/p$, then the infimum of $\psi$ occurs at $\zeta=2a$ and its value is $\psi(a,2a) = \psi_1(2a)=1/2$.
The reason is that on $[a,2-1/p]$, $\psi = \psi_1$
achieves its minimum at $2a$,
and on $[2-1/p,1]$, $\psi = \psi_2$ is increasing since
$\Phi(a,2-1/p) \ge 0$.

(ii) If $a \le 2 - 1/p$ and $2a > 2 - 1/p$, then the infimum occurs at $\zeta =\phi(a)$ and its value is
$\frac{\phi(a)-a}{\phi(a)} \left( (1-p) (2-\phi(a)) / (1-\phi(a)) \right)^a$.
The reason is that on $[a,2-1/p]$, $\psi = \psi_1$
is decreasing, and on $[2-1/p,1]$, $\psi = \psi_2$
achieves its minimum at $\phi(a)$ since $\Phi(a,2-1/p)\le0$ and $\Phi(a,1)\ge0$.

(iii) If $a > 2 - 1/p$, then the infimum occurs at $\zeta =\phi(a)$ and its value is equal to
$\frac{\phi(a)-a}{\phi(a)} \left( (1-p) (2-\phi(a)) / (1-\phi(a)) \right)^a$.
The reason is that on $[a,1]$,
$\psi=\psi_2$ achieves its minimum  at $\phi(a)$ since $\Phi(a,a)\le0$ and $\Phi(a,1)\ge0$.

\noindent\textbf{Case 2: $p < 1/2$.}
By the definition of $h$ in \eq{h_def} we have
$$
\psi(a,\zeta) =
\begin{cases}
\left(\frac{1-p}{p}\right)^{a}
\frac{\zeta-a}{\zeta} \left(\frac{a}{\zeta-a}\right)^{a/\zeta}
& \mathrm{\ if\ } a \le \zeta \le \frac{1-2p}{1-p} \\
\frac{\zeta-a}{\zeta} \left(\frac{a}{(\zeta-a)f(\zeta)}\right)^{a/\zeta}
& \mathrm{\ otherwise.}
\end{cases}
$$
The function $\psi$ is continuous here since $$f\left(\frac{1-2p}{1-p}\right)=\left(\frac{p}{1-p}\right)^{\frac{1-2p}{1-p}} \:.$$
Let us define
$\psi_3(\zeta) = \left(\frac{1-p}{p}\right)^{a}
\frac{\zeta-a}{\zeta} \left(\frac{a}{\zeta-a}\right)^{a/\zeta}
$.
Since $\psi_3(\zeta) = \left(\frac{1-p}{p}\right)^{a} \psi_1(\zeta)$,
the function $\psi_3(\zeta)$ is decreasing when $\zeta \le 2a$ and increasing when $\zeta\ge 2a$.
So $\psi_3$ achieves its minimum at $\zeta = 2a$ and its minimum value is $\left(\frac{1-p}{p}\right)^{a}/2$.
We conclude that

(iv) If $a \le 1 - p/(1-p)$ and $2a \le 1 - p/(1-p)$,
then the infimum in \eq{ntsgU} occurs
at $\zeta=2a$ and at this point we have
$\psi(a,\zeta) = \left(\frac{1-p}{p}\right)^a/2$.
The reason is that
on $[a, 1 - p/(1-p)]$, $\psi = \psi_3$ achieves its minimum at $2a$,
and on $[1-p/(1-p),1]$, $\psi = \psi_2$ is increasing since
$\Phi(a,1 - p/(1-p))\ge 0 $.

(v) If $a \le 1 - p/(1-p)$ and $2a > 1 - p/(1-p)$,
then the infimum in \eq{ntsgU} occurs
at $\zeta =\phi(a)$ and its value is equal to
$\frac{\phi(a)-a}{\phi(a)} \left( (1-p) (2-\phi(a)) / (1-\phi(a)) \right)^a$.
The reason is that on $[a, 1 - p/(1-p)]$, $\psi = \psi_3$ is decreasing,
and on $[1-p/(1-p),1]$, $\psi = \psi_2$ achieves its minimum at $\phi(a)$
since $\Phi(a,1 - p/(1-p))\le0$ and $\Phi(a,1)\ge0$.

(vi) If $a > 1 - p/(1-p)$,
then the infimum in \eq{ntsgU} occurs
at $\zeta =\phi(a)$ and its value is equal to
$\frac{\phi(a)-a}{\phi(a)} \left( (1-p) (2-\phi(a)) / (1-\phi(a)) \right)^a$.
The reason is that on $[a,1]$, $\psi = \psi_2$ achieves its minimum
at $\phi(a)$ since
$\Phi(a,a)\le 0 $ and $\Phi(a,1) \ge 0$.

In all cases we proved that $g_U(a)$ actually gives the value of the infimum in \eq{ntsgU}, and this concludes the proof of \eq{ntsgU} and of part (a).

(b)
Consider the definition of $g_U$ in \eq{gu}.
The formulae in \eq{gUder} for the cases  
`$p\ge 1/2$ and $0 < a \le 1 - 1/2p$'
and 
`$p < 1/2$ and $0 < a \le \frac{1-2p}{2-2p}\:$'
are clearly true, so we assume that $a$ is in the `otherwise' case.
We use the equality~\eq{ntsgU}.
Note that as proved in part (a), the infimum in~\eq{ntsgU} occurs at the point $\zeta=\phi(a)$ that has $\left.\frac{\partial \psi}{\partial \zeta}\right|_{(a,\phi(a))} = 0$. This implies for every $a_0$,
\begin{align*}
\frac{\mathrm{d} g_U}{\mathrm{d} a}(a_0) & =
\frac{\partial \psi}{\partial a} {(a_0,\phi(a_0))}
+
\frac{\partial \psi}{\partial \zeta} {(a_0,\phi(a_0))} \times
\frac{\mathrm{d} \phi}{\mathrm{d} a}({a_0}) \\
&
=\frac{\partial \psi}{\partial a} {(a_0,\phi(a_0))}
=\left.
\frac{\partial}{\partial a}
\left[
\frac{\zeta-a}{\zeta} \left(\frac{a}{(\zeta-a)f(\zeta)}\right)^{a/\zeta}\right]
\right|_{(a_0,\phi(a_0))} \:,
\end{align*}
and \eq{gUder} follows from computing this partial derivative
and putting $\zeta = \phi(a_0)$.

We next prove the continuity of $g_U$ and its derivative.
Note that by Lemma~\ref{lem:s}(a), if $a\in(0,1)$ then $\phi(a)\in(0,1)$.
First, $g_U$ is continuous at $a=1$ since
\begin{align*}
\lim_{a\to 1}
\frac{\phi(a)-a}{\phi(a)} \left( \frac{(1-p) (2-\phi(a))}{1-\phi(a)} \right)^a &
=
\lim_{a\to 1}
\frac{\phi(a)-a}{\phi(a)} \left( \frac{a}{p(2-\phi(a))(\phi(a)-a)} \right)^a \\
&
= \lim_{a\to 1}
\frac{(\phi(a)-a)^{1-a}}{\phi(a)} \left( \frac{a}{p(2-\phi(a))} \right)^a
=
\frac{1}{p} \:.
\end{align*}

For $p\ge 1/2$, the only discontinuity for $g_U$ can possibly occur at $b = 1 - 1/2p$. However at this point we have $\phi(b) = 2b = 2 - p^{-1}$
so that $(1-p)(2-\phi(b)) = 1-\phi(b)$.
Hence the left and right limits of $g_U$ equal $1/2$,
and the left and right limits of $g'_U$ equal 0.
Therefore, both $g_U$ and $g'_U$ are continuous at $b$.

For $p < 1/2$, the only discontinuity for $g_U$ can possibly occur at $c = (1-2p)/(2-2p)$. However at this point $\phi(c) = 2c = (1-2p)/(1-p)$
so that $\frac{(1-p)(2-\phi(c))}{1-\phi(c)} = \frac{1-p}{p}$.
Hence the left and right limits of $g_U$ equal $\left(p^{-1}-1\right)^{c}/2$,
and the left and right limits of $g'_U$ equal $\log\left(p^{-1}-1\right)\left(p^{-1}-1\right)^{c}/2$.
Therefore, both $g_U$ and $g'_U$ are continuous at $c$.

(c)
Note that $g_U$ is positive everywhere, so $\log(g_U)$ is (strictly) increasing
if and only if $g_U$ is (strictly) increasing.
By the formulae for $g'_U$ in part (b),
it is easy to see that $g'_U$ is always non-negative,
and is positive when $g_U(a) > 1/2$.
To show $\log (g_U)$ is convex, we need to show its derivative, i.e. $g'_U/g_U$ is increasing.
This also follows from part (b), noting that $\phi$ is increasing by Lemma~\ref{lem:s}(c).
\end{proof}

\end{document}